\documentclass[11pt]{article}

\usepackage{graphicx, amsmath,amsthm,amssymb,url}

\usepackage[labelfont=bf]{subcaption}
\usepackage[labelfont=bf]{caption}
\usepackage{algorithm,algcompatible,mathtools}
\usepackage[multidot]{grffile}
\usepackage[normalem]{ulem} % just for strikeout
\usepackage[breaklinks=true]{hyperref}
\usepackage{fullpage,etoolbox}
\usepackage{color}
\usepackage[numbers,sort,compress]{natbib}
\usepackage{authblk}

\newtheorem{theorem}{Theorem}
\numberwithin{theorem}{section}
\newtheorem{lemma}[theorem]{Lemma}

\newtheorem{proposition}[theorem]{Proposition}
\newtheorem{coro}[theorem]{Corollary} 
\newtheorem{remark}{Remark}

\numberwithin{equation}{section}

\DeclareMathOperator*{\minimize}{minimize \quad}

\DeclareMathOperator*{\st}{subject\,\, to \quad}

\newcommand{\R}{\mathbb{R}}

\newcommand{\RR}{\ensuremath{\mathbb{R}}}
\newcommand{\cbf}{h}

\newcommand{\statedim}{n}
\newcommand{\inputdim}{m}
\newcommand{\xx}{x}

\newcommand{\uu}{u}
\newcommand{\g}{g}
\newcommand{\q}{q}
\newcommand{\f}{f}
\newcommand{\F}{\ensuremath{\mathcal{F}}}
\newcommand{\Ss}{\ensuremath{\mathcal{S}}}

\newcommand{\D}{\ensuremath{\mathcal{D}}}
\newcommand{\bD}{\ensuremath{{\bar{\mathcal{D}}}}}
\newcommand{\C}{\ensuremath{\mathcal{C}}}

\newcommand{\N}{\ensuremath{\mathcal{N}}}

\newcommand{\I}{\ensuremath{\mathcal{I}}}
\newcommand{\U}{\ensuremath{\mathcal{U}}}
\newcommand{\B}{\ensuremath{\mathcal{B}}}
\newcommand{\Lc}{L_h}
\newcommand{\Ln}{L_\nabla}
\newcommand{\Lq}{L_q}
\DeclareMathOperator*{\bd}{bd}

\DeclareMathOperator*{\inte}{int}
\newcommand{\norm}[1]{\| #1 \|}

\newcommand{\Lip}{\mathrm{Lip}}

\newcommand{\Xsafe}{X_{\mathrm{safe}}}
\newcommand{\bXsafe}{\bar{X}_{\mathrm{safe}}}
\newcommand{\Xunsafe}{X_{\mathrm{unsafe}}}

\newcommand{\gammasafe}{\gamma_{\mathrm{safe}}}
\newcommand{\gammaunsafe}{\gamma_{\mathrm{unsafe}}}

\newcommand{\XN}{X_{\mathcal{N}}}
\newcommand{\XS}{X_{\mathcal{S}}}
\newcommand{\Zdynamics}{Z_{\mathrm{dyn}}}
\newcommand{\gammadynamics}{\gamma_{\mathrm{dyn}}}

\newcommand{\calH}{\mathcal{H}}

\newcommand{\ip}[2]{\langle #1, #2 \rangle}
\newcommand{\T}{\mathsf{T}}
\newcommand{\abs}[1]{| #1 |}

\title{Learning Control Barrier Functions from Expert Demonstrations}
\author[1]{Alexander Robey\thanks{A.\ Robey and H.\ Hu contributed equally.}}
\author[1]{Haimin Hu$^{\ast}$}
\author[2]{Lars Lindemann}
\author[1]{Hanwen Zhang}
\author[2]{Dimos V.\ Dimarogonas}
\author[3]{Stephen Tu}
\author[1]{Nikolai Matni}
\affil[1]{Department of Electrical and Systems Engineering, University of Pennsylvania}
\affil[2]{Division of Decision and Control Systems,  KTH Royal Institute of Technology}
\affil[3]{Google Brain Robotics}

\date{April 8, 2020, Revised: \today}

\begin{document}
\maketitle

\begin{abstract}
Inspired by the success of imitation and inverse reinforcement learning in replicating expert behavior through optimal control, we propose a learning based approach to safe controller synthesis based on control barrier functions (CBFs).  We consider the setting of a known nonlinear control affine dynamical system and assume that we have access to safe trajectories generated by an expert — a practical example of such a setting would be a kinematic model of a self-driving vehicle with safe trajectories (e.g., trajectories that avoid collisions with obstacles in the environment) generated by a human driver.  We then propose and analyze an optimization-based approach to learning a CBF that enjoys provable safety guarantees under suitable Lipschitz smoothness assumptions on the underlying dynamical system.  A strength of our approach is that it is agnostic to the parameterization used to represent the CBF, assuming only that the Lipschitz constant of such functions can be efficiently bounded. Furthermore, if the CBF parameterization is convex, then under mild assumptions, so is our learning process.  We end with extensive numerical evaluations of our results on both planar and realistic examples, using both random feature and deep neural network parameterizations of the CBF.  To the best of our knowledge, these are the first results that learn provably safe control barrier functions from data.
\end{abstract}
\section{Introduction}

Consider the following safety-critical scenarios: a self-driving car navigating through traffic, two unmanned aerial vehicles (UAVs) avoiding collision, and a robotic manipulator in a laboratory setting that must avoid injuring researchers.  Although vastly different in terms of their environments, safety-specifications, and underlying dynamics, they share several key properties: (i) their dynamics are well understood and modeled, and can be accurately identified, (ii) their dynamics are inherently \emph{nonlinear}, and (iii) \emph{expert demonstrations} of safe and desirable behavior are readily available or can be easily collected.  Motivated by these unifying properties, this paper proposes the design of safe controllers for known nonlinear dynamical systems based on \emph{control barrier functions} learned from expert demonstrations.

Barrier functions, which are also referred to as barrier certificates, were first proposed in \cite{prajna2007framework} as a means of certifying the safety of dynamical systems with respect to semi-algebraic safe sets.  In that work, a sum-of-squares (SOS) programming \cite{parrilo2000structured} approach for synthesizing polynomial barrier functions for given polynomial systems was also described. The notion of control barrier functions (CBFs) for dynamical control systems was first introduced in \cite{wieland2007constructive} to guarantee the existence of a control law that renders a desired safe set forward invariant. The notion of CBFs was refined by introducing reciprocal \cite{ames2014control} and zeroing CBFs \cite{ames2017control}, which do not require that sub-level sets of the CBF be invariant within the safe set. In particular, zeroing CBFs can be used to compute a minimally invasive ``correction'' to a nominal control law.  Importantly, this correction maintains safety by computing the solution to a quadratic program (QP) \cite{ames2017control}. 

One open problem that has not been fully addressed in prior work is how such CBFs can be synthesized for general classes of systems.  This challenge is similar to that which arises when addressing stability using control Lyapunov functions (CLFs) as the analog to Lyapunov functions \cite{sontag1989universal}. Notably, control Lyapunov functions are a subset of control barrier functions (see \cite{ames2017control} and \cite{xu2015robustness}). Analytic and SOS based approaches to synthesizing CBFs and CLFs are summarized in \cite{ames2019control} and have appeared in \cite{xu2017correctness,wang2018permissive}. These approaches, however, are known to be limited in scope and scalability.

\subsection{Related work on learning and CBFs}
Methods using barrier and control barrier functions to ensure safety and guide exploration during episodic supervised learning of uncertain linear dynamics include
\cite{cheng2019end,wang2018safe,taylor2019learning,taylor2020control}.  These approaches typically assume that a valid (control) barrier function is provided, and should be viewed as complementary to our results.  In \cite{yaghoubi2020training}, an imitation learning based approach is used to to train a deep neural network (DNN) to replicate a CBF based controller.  While the authors of \cite{yaghoubi2020training} present empirical validation of their results, no theoretical guarantees of correctness are provided. The authors of \cite{jin2020neural} jointly learn a control Lyapunov function, a CBF, and
a policy function for which their validity is then verified post-hoc using Lipschitz arguments. The authors of \cite{jin2020neural} jointly learn a control Lyapunov function, a CBF, and
a policy function, and then verify their validity post-hoc using Lipschitz arguments. In \cite{boffi2020learning}, tools from statistical learning theory, are used to learn Lyapunov functions from data for systems with unknown dynamics.    Most similar in spirit to our paper are the results in \cite{srinivasan2020synthesis} and \cite{saveriano2019learning}.  In \cite{srinivasan2020synthesis}, the authors parameterize a CBF by a support vector machine, and use a supervised learning approach to characterize regions of the state-space as safe or unsafe based on collected data.  
While conceptually appealing, we note that their training procedure does not ensure \emph{a priori} that there exist control actions such that the learned safe set can be made forward invariant,\footnote{In particular, they do not ensure that the derivative condition  $\langle\nabla h(x),f(x,u)\rangle + \alpha(h(x))\geq 0$, holds for the learned CBF $h(x)$ at the observed data points, with $f(x,u)$ the system dynamics, and $\alpha$ an extended class $\mathcal K$ function -- see Section \ref{sec:problem} for more details.}  and hence cannot guarantee safe execution of the system.  
In \cite{saveriano2019learning}, a method is proposed which incrementally learns a \emph{linear} CBF by clustering expert demonstrations into linear subspaces and fitting low dimensional representations.  While both papers \cite{srinivasan2020synthesis,saveriano2019learning} empirically validate their methods, neither provide proofs of correctness of the learned CBF.

\vspace{10pt}

\textbf{Contributions.}  In this paper, we propose and analyze an optimization based approach to learning a zeroing CBF (henceforth referred to simply as a CBF) from expert trajectories for known control affine nonlinear systems.  In particular, we provide precise and verifiable conditions on the expert trajectories, an additional auxiliary data-set, and the hyperparameters of the optimization problem so as to ensure that the learned CBF guarantees safe execution of the system.  We further show how the underlying optimization problem can be efficiently solved when it is cast over different function spaces.  In particular, we show that the problem can be solved via convex optimization when the function space lies within a (possibly infinite-dimensional) reproducing kernel Hilbert space (RKHS); alternatively, when we consider the function space of deep neural networks (DNNs), the problem can be solved via first-order stochastic methods such as Adam or SGD.  To the best of our knowledge, these are the first such results that learn a CBF from expert demonstrations with provable safety guarantees.

\vspace{10pt}

\textbf{Paper structure.}  The rest of this paper is structured as follows.  In Section \ref{sec:problem}, we introduce notation and formulate the general problem of learning a CBF from expert demonstrations.  In Section \ref{sec:optim}, we derive a set of sufficient conditions on the learned CBF and data-set that guarantee safety of the resulting closed-loop system, and we subsequently use these conditions to formulate an optimization problem for computing a function satisfying these conditions.  We show in Section \ref{sec:compute} that this optimization problem can be efficiently solved for CBFs embedded in RKHS and DNN function classes, and in Section \ref{sec:data}, we provide further details on the expert trajectory collection process.  We present three numerical studies in Section \ref{sec:experiments}: (i) a two-dimensional planar problem for which we explicitly compute and verify all of the conditions of our main theorem, showing that the conditions are indeed satisfied in practice, (ii) a two UAV collision-avoidance example where expert trajectories are generated by the closed form CBF from \cite{squires2018constructive}, and (iii) the same two UAV collision avoidance example, where now expert trajectories are generated by human players of a video game interface.  We end with conclusions and discussions of directions for future work in Section \ref{sec:conclusion}.
%Our proposed methods is generic and extends in a straightforward manner to all aforementioned extensions of zeroing barrier
\section{Preliminaries and problem formulation}
\label{sec:problem}
Let $\RR$ and $\RR_{\ge 0}$ be the set of real and non-negative real numbers, respectively, and $\RR^\statedim$  the set of $\statedim$-dimensional real vectors. For $\epsilon>0$ and $p\ge 1$, we let $\B_{\epsilon,p}(\bar{\xx}):=\{\xx\in\RR^\statedim\, \big{|}\, \|\xx-\bar{\xx}\|_p\le \epsilon\}$ denote the closed $p$-norm ball around $\bar{\xx}\in\RR^\statedim$.  For a given set $\C$, we denote by $\bd(\C)$, $\inte(\C)$, and $\C^c$  the boundary, interior, and complement of $\C$, respectively. For two sets $\C_1$ and $\C_2$, we denote their Minkowski sum by $\C_1\oplus\C_2:=\{\xx_1+\xx_2\in\RR^\statedim|\xx_1\in\C_1,\xx_2\in\C_2\}$. A continuous function $\alpha:\RR\to\RR$ is an extended class $\mathcal{K}$ function if it is strictly increasing with $\alpha(0)=0$.  The inner-product between two vectors $x,y\in\RR^n$ is denoted by $\langle x,y\rangle$.

\subsection{Valid control barrier functions}
At time $t\in\RR_{\ge 0}$, let $\xx(t)\in\RR^\statedim$ and $\uu(t)\in\RR^\inputdim$ be the state and input, respectively, of the dynamical control system  described by the initial value problem
\begin{align}\label{eq:system}
\dot{\xx}(t)=\f(\xx(t))+\g(\xx(t))\uu(t), \quad \xx(0)\in\RR^\statedim
\end{align}
where $\f:\RR^\statedim\to\RR^\statedim$ and $\g:\RR^\statedim\to\RR^\inputdim$ are locally Lipschitz continuous functions. Let the unique solution to \eqref{eq:system} under a locally Lipschitz continuous control law $\uu:\RR^\statedim\to\RR^\inputdim$ be $\xx:\I\to \RR^\statedim$ where $\I\subseteq\RR_{\ge 0}$ is the maximum definition interval of $\xx$. Note that we do not explicitly assume forward completeness of \eqref{eq:system} under $\uu$ here, i.e., $\I$ may be bounded.

Consider next a {twice} continuously differentiable function $\cbf:\RR^\statedim\to\RR$, and define the set 
\begin{align}\label{eq:set_C}
\C:=\{\xx\in\RR^\statedim \, \big{|} \, \cbf(\xx)\ge 0\},
\end{align}
as the set that we wish to certify as safe, i.e., the set $\C$ satisfies prescribed safety specifications and can be made forward invariant through an appropriate choice of control action. We further assume that $\C$ has non-empty interior, and let $\D$ be an open set such that $\D\supset \C$.  The function $\cbf(\xx)$ is said to be a \emph{valid control barrier function} on $\D$ if there exists a locally Lipschitz continuous extended class $\mathcal{K}$ function $\alpha:\RR\to\RR$ such that 
\begin{align}\label{eq:cbf_const}
\sup_{\uu\in \U} %\frac{\partial \cbf(\xx)}{\partial \xx}
\big\langle \nabla \cbf(\xx), \f(\xx)+\g(\xx)\uu\big\rangle \ge -\alpha(\cbf(\xx))
\end{align} 
holds for all $\xx\in\D$, where $\U\subset\RR^\inputdim$ defines constraints on the control input $u$. Consequently, we define the set of \emph{CBF consistent inputs} induced by a valid CBF $h(x)$ to be 
\[
K_{\text{CBF}}(\xx):=\{\uu\in\RR^\inputdim \, \big{|} \, %\frac{\partial \cbf(\xx)}{\partial \xx}
\langle \nabla \cbf(\xx), \f(\xx)+\g(\xx)\uu\rangle \ge -\alpha(\cbf(\xx))\}.\]

The next result follows from \cite{ames2017control,xu2015robustness}.
\begin{lemma}
Assume that $\cbf(\xx)$ is a valid control barrier function on $\D$ and that $\uu:\D\to \U$ with $\uu(\xx)\in K_{\text{CBF}}(\xx)$ is locally Lipschitz continuous. Then it holds that $\xx(0)\in\C$ implies $\xx(t)\in\C$ for all $t\in \I$. If the set $\C$ is compact, it additionally follows that $\C$ is: 1) forward invariant, i.e., $\I=[0,\infty)$, and 2) asymptotically stable, which implies that $\xx(t)$ approaches $\C$ as $t\to\infty$ when $\xx(0)\in\C^c \cap \D$.
\end{lemma}

Note that $\cbf(\xx)\neq 0$ when $x\in\bd(\C)$ (see \cite[Remark 5]{ames2019control}) is not required when using the Comparison Lemma instead of Nagumo's theorem to prove the above result. While the previous result provides strong guarantees of safety given a valid control barrier function, one is still left with the potentially daunting task of finding a twice continuously differentiable function $h$ such that (i) the set $\C$ defined in equation \eqref{eq:set_C} captures a sufficiently large volume of ``safe'' states needed for the task at hand, and (ii) that it satisfies the derivative constraint \eqref{eq:cbf_const} on an open set $\D\supseteq\C$.  While safety constraints are often naturally specified on a subset of the configuration space of a system (e.g., to avoid collision, vehicles must maintain a minimum separating distance), ensuring that a CBF specified using such geometric intuition also satisfies constraint \eqref{eq:cbf_const} can involve verifying complex relationships between the vector field of the system, the candidate control barrier function, and its gradient. 

As described in the introduction, this challenge motivates the approach taken in this paper, wherein we propose an optimization based approach to learning a CBF from expert demonstrations for a system with known dynamics.  

\subsection{Problem formulation}
To formalize the previous discussion, we explicitly distinguish between geometric safety specifications, i.e., those that can be directly specified on (a subset) of the state-space of the system $x \in \RR^n$, and the set $\C$ defined in equation \eqref{eq:set_C} that is certified as safe by the CBF.  To that end, let $\Ss\subseteq \RR^\statedim$ define the aforementioned geometric safe set.  

Toward the goal of learning a valid CBF, we assume that we are given a set of \emph{expert trajectories}\footnote{We refer to the collection of data points $\Zdynamics$ as expert trajectories to emphasize that this is a natural way of collecting the $\{(x_i,u_i)\}$ pairs from the system \eqref{eq:system}.  We note however that our method simply requires a collection of state-action pairs $\{(x_i,u_i)\}$ demonstrating safe behavior, and that they need not arise from sequential sampling of expert trajectories.} consisting of $N_1$ discretized data-points $\Zdynamics:=\{(\xx_i,\uu_i)\}_{i=1}^{N_1}$ such that $\xx_i\in\inte(\Ss)$.  This is illustrated in Figure \ref{fig:exp-trajs}. For $\epsilon>0$, we define the sets
\begin{align}\label{eq:set_D}
    \D':=\bigcup_{i=1}^{N_1}\B_{\epsilon,p}(\xx_i)\qquad\text{and}\quad \ \D := \D' \backslash \mathrm{bd}(\D')
\end{align} 
where $\mathcal{D}$ needs to be such that  $\mathcal{D}\subseteq \mathcal{S}$ to later ensure correctness of the learned CBF. This can be easily achieved even when data-points $x_i$ are close to $\text{bd}(\mathcal{S})$  by adjusting $\epsilon$ or by omitting $x_i$. Several comments are in order.  First, note that we define $\D$ based on expert trajectories for which control inputs $\uu_i$ are available so that the derivative constraint \eqref{eq:cbf_const} can be enforced during learning.  Second, by construction, the $x$ component of $\Zdynamics$ defines an $\epsilon$-net over $\D$, i.e., for all $\xx \in \D$, (slightly abusing notation) there exists $x_i\in\Zdynamics$ such that $\|x_i-x\|_p\leq \epsilon$.  Finally, conditions on $\epsilon$ will be specified later to ensure the validity of the learned CBF.

\begin{remark}
We note that a conceptually similar approach, defined in terms of taking a point-wise union over previously seen safe trajectories, is used to define a safe terminal set in the Learning Model Predictive Control method of \cite{rosolia2017learning}.
\end{remark}

We next define the set $\N$, for $\sigma>0$, as
\begin{align*}
    \N:=\{\bd(\D)\oplus\B_{\sigma,p}(0)\} \setminus \D,
\end{align*} 
which should be thought of as a ``layer'' of width $\sigma$ surrounding the set $\D$; see Figure \ref{fig:learned-safe-set} for a graphical depiction.  As will be made clear in the sequel, by enforcing that the value of the learned CBF $\cbf(\xx)$ is negative on the set $\N$, which can be accomplished through appropriate sampling, we ensure that the zero level set $\{ \xx \in \RR^n \, | \, \cbf(\xx) = 0 \}$ is contained within the set $\D$, which is a necessary condition for $\cbf(\xx)$ to be valid.

While the above definition of a CBF is specified over all of $\RR^n$, e.g., the definition of the set $\C$ in equation \eqref{eq:set_C} considers all $x\in \RR^n$ such that $h(x) \geq 0$, we make a minor modification to this definition in order to restrict the domain of interest to the set $\N\cup\D$, i.e., we will certify that $h(x)$ is a valid \emph{local} CBF over the set $\D$ with respect to the set
\begin{align} \label{eq:local_C}
\C:=\{\xx\in\N\cup\D\, \big{|} \, \cbf(\xx)\ge 0\}.
\end{align}

This restriction is natural, as we are learning a CBF $h(x)$ from data sampled only over the domain $\N\cup\D$, and we will show that the inclusion $\C\subset\D\subseteq\Ss$ holds.  It then follows that if $\cbf(\xx)$ is shown to satisfy the derivative constraint \eqref{eq:cbf_const} for all $\xx \in \D$, then both the set $\C$, as defined in \eqref{eq:local_C}, and the set $\D$ can be made forward invariant by some $u\in K_{\text{CBF}}(\xx)$, i.e., by some control action $u\in\U$ satisfying the derivative condition \eqref{eq:cbf_const} with respect to the learned CBF $h(x)$. %$u\in \U$.

\begin{figure}
    \centering
    \begin{subfigure}[t]{0.3\textwidth}
        \includegraphics[width=\textwidth]{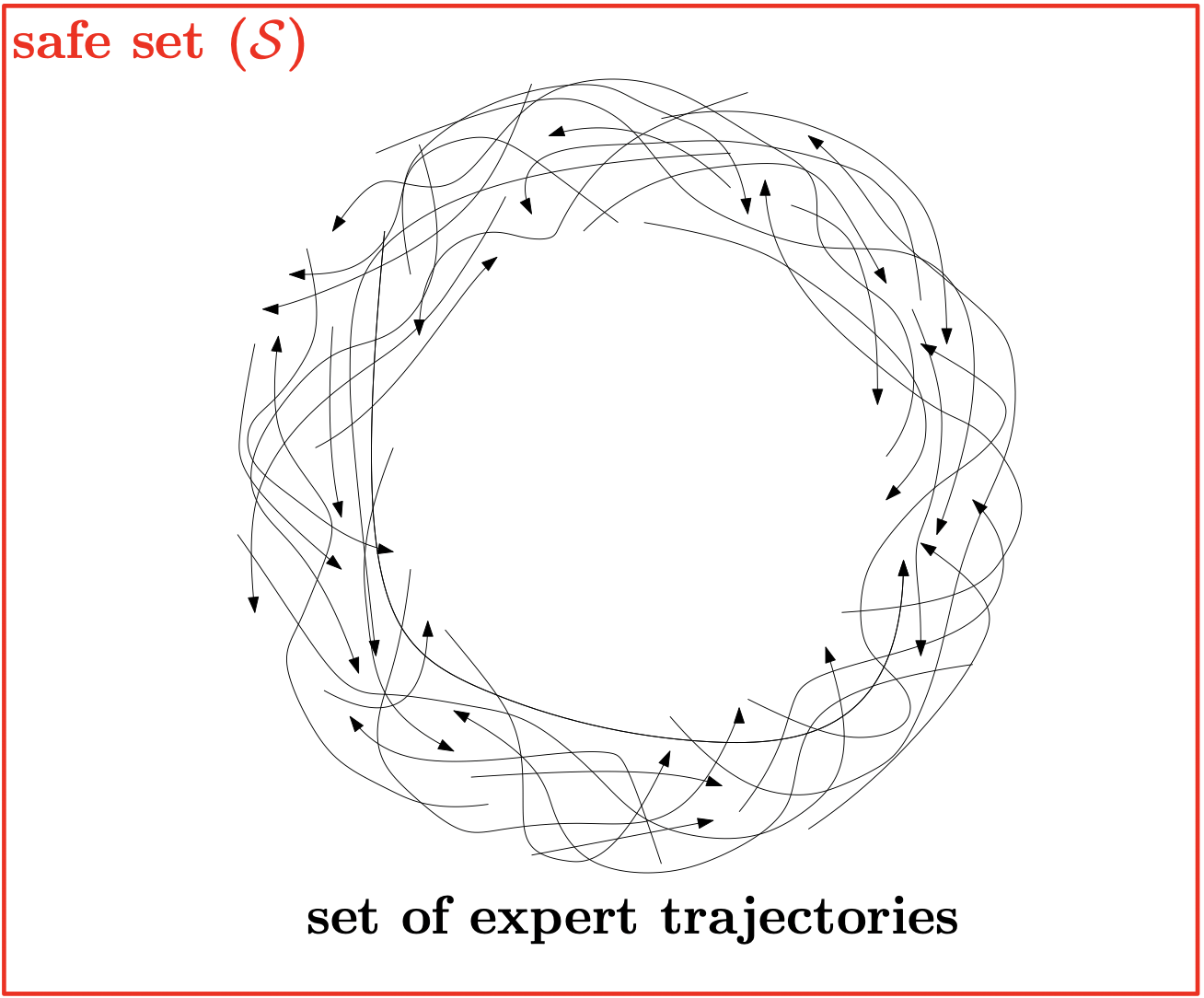}
        \caption{\textbf{Problem setup.}}
        \label{fig:exp-trajs}
    \end{subfigure}\quad 
    \begin{subfigure}[t]{0.3\textwidth}
        \includegraphics[width=\textwidth]{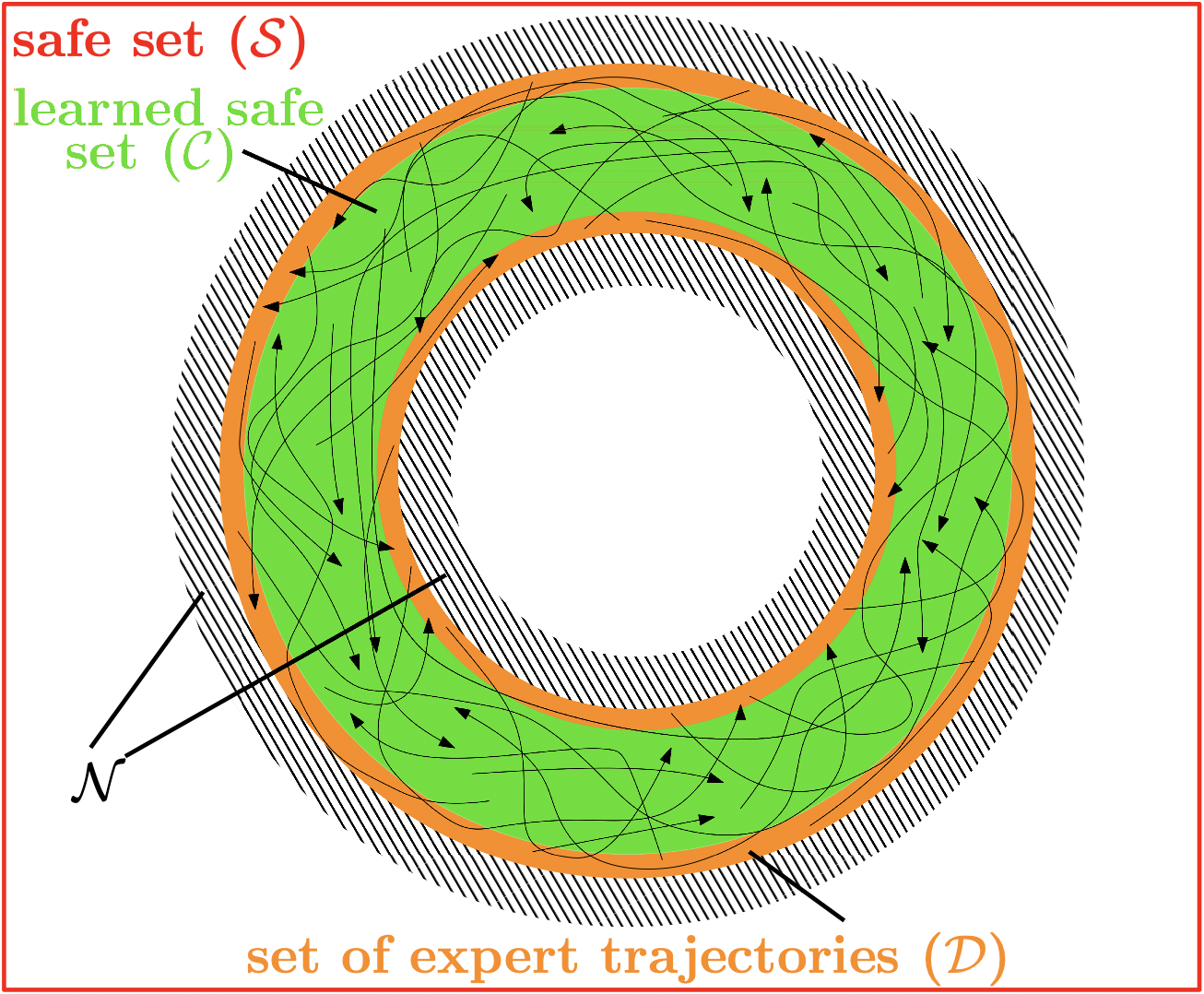}
        \caption{\textbf{Desired result.}}
        \label{fig:learned-safe-set}
    \end{subfigure}\quad 
    \begin{subfigure}[t]{0.3\textwidth}
        \includegraphics[width=\textwidth]{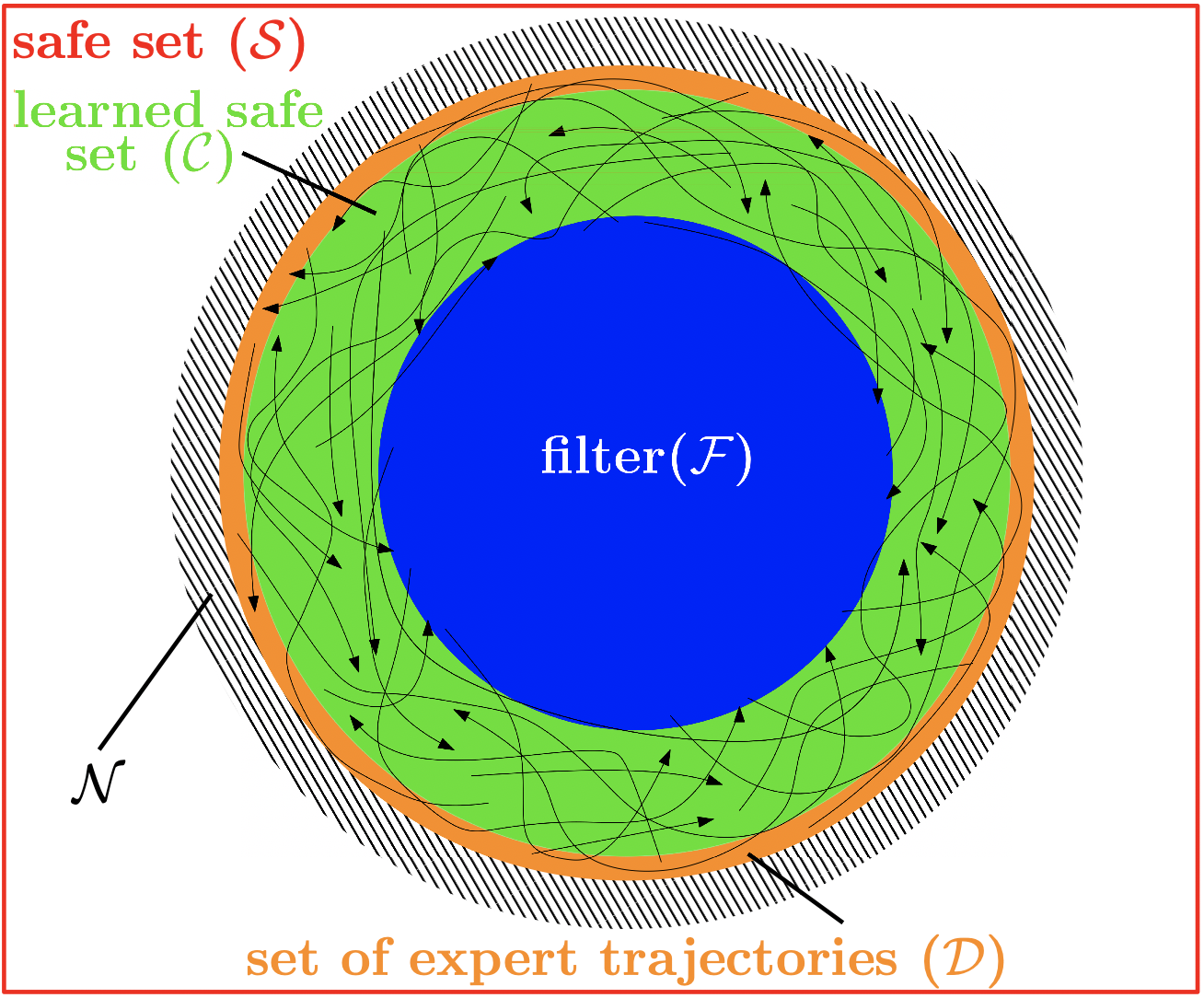}
        \caption{\textbf{Control barrier filter.}}
        \label{fig:control-barrier-filter}
    \end{subfigure}
    \caption{In (a), the safe set $\Ss$ (red box) and the set of expert trajectories (black lines).  Next, in (b), the set $\D$ (orange ring) is the union of $\epsilon$ balls around the expert trajectories. The set $\N$ (black striped rings), defined around $\D$, ensures that the learned safe set $\C$ (green ring), which is defined via the learned valid control barrier function $\cbf(\xx)$, is such that $\C\subset\D\subseteq\Ss$.  Finally, in (c), ``artiticial'' unsafe samples are no longer introduced in the center of the safe set (denoted by the blue set $\F$).}
    \label{fig:1}
\end{figure}

\section{An optimization based approach}
\label{sec:optim}

In this section, we define and analyze an optimization based approach to {synthesizing}
%generating
valid local control barrier functions from expert demonstrations.  To this end, let $\calH$ be a normed function space of twice continuously differentiable 
functions $h : \RR^n \to \RR$ {for which local Lipschitz bounds 
\begin{align*}
    \Lc(x) := \sup_{x_1, x_2 \in \B_{\epsilon,p}(x)} \frac{\abs{h(x_1) - h(x_2)}}{\norm{x_1 - x_2}_p}
\end{align*}
can be efficiently estimated}.  Commonly used examples of such spaces include infinite dimensional reproducing kernel Hilbert spaces (RKHS) such as those defined by random Fourier (RF) features \cite{ rahimi2008random}, and more recently deep neural networks (DNNs) \cite{fazlyab2019efficient}.  We defer a discussion of results specific to these two classes of CBFs to the end of this section, and focus now on a general method applicable to these, and other, spaces $\calH$.

Recall the definition of $\Zdynamics$ and define the set $\Xsafe = \{ x_i : (x_i, u_i) \in \Zdynamics \}$. 
We also assume that points $\XN = \{\xx_i\}_{i=1}^{N_2}$ are sampled from the set $\N$ such that $\XN$ forms an $\bar{\epsilon}$-net of $\N$ -- conditions on $\bar{\epsilon}$ will be specified in the sequel. We emphasize that no associated inputs $\uu_i$ are needed for the samples $\XN \subset \N$, as these points are not generated by the expert, and can instead be obtained by simple computational methods such as gridding or uniform sampling.

We begin by deriving a set of sufficient conditions in terms of constraints on the learned CBF $\cbf(\xx)$, as well as conditions on the data-sets $\Xsafe$ and $\XN$, that ensure that $\cbf(\xx)$ is a valid local CBF on $\D$.  We then use these constraints to formulate an optimization problem that can be efficiently solved for the aforementioned function classes $\calH$.

\subsection{Guaranteeing $\C\subset\D\subseteq\Ss$}

We begin with the simple and intuitive requirement that the learned CBF $\cbf(\xx)$ satisfy
\begin{align}\label{eq:safe}
    \cbf(\xx_i)\geq\gammasafe \quad\forall x_i \in \Xsafe,
\end{align}
for a yet to be specified parameter $\gammasafe>0$.  This in particular ensures that the set $\C$ over which $\cbf(\xx)\geq 0$, as defined in equation \eqref{eq:local_C}, has non-empty interior. 

We now derive conditions under which the learned CBF satisfies $h(x) < 0$ for all $ x \in \N$,
which in turn ensures that $\C\subset\D\subseteq\Ss$ due to constraint \eqref{eq:safe}.

\begin{proposition}\label{thm:0}
Let $\cbf(\xx)$ be Lipschitz continuous with local Lipschitz constant $\Lc(\xx)$.
Let $\gammaunsafe>0$ and $\XN$ be an $\bar{\epsilon}$-net of $\N$ with $\bar{\epsilon}<\gammaunsafe/\Lc(x_i)$ for all $x_i \in \XN$. Then, if 
\begin{align}\label{eq:unsafe}
    \cbf(\xx_i)\leq-\gammaunsafe \quad\forall x_i \in \XN
\end{align} 
it holds that $\cbf(\xx)<0$ for all $x\in\N$.
\end{proposition}
\begin{proof} 
By equation \eqref{eq:unsafe}, we have that $\cbf(\xx_i)\leq-\gammaunsafe$ for each $\xx_i\in\XN$. We then have, for any $x\in\N$, that there exists a point $x_i \in \XN$ satisfying $\norm{x-x_i}_p\leq \bar{\epsilon}<\gammaunsafe/L(x_i)$, from which the following chain of inequalities follows immediately
\begin{align*}
    \cbf(\xx) &=  \cbf(\xx)-\cbf(\xx_i) + \cbf(\xx_i) \leq |\cbf(\xx)-\cbf(\xx_i)| - \gammaunsafe  \\
    & \leq \Lc(x_i)\norm{x-x_i}_p - \gammaunsafe \leq \Lc(x_i)\bar{\epsilon} - \gammaunsafe < 0
\end{align*}
where the first inequality follows from the assumption that $\cbf(\xx_i)\leq - \gammaunsafe$ for all $\xx_i \in \XN$, the second by the local Lipschitz assumption on $\cbf(x)$, the third by the assumption that $\XN$ forms an $\bar{\epsilon}$-net of $\N$, and the final inequality by the condition on $\bar{\epsilon}$ of the proposition.
\end{proof}

We note that as stated, the constraints \eqref{eq:safe} and \eqref{eq:unsafe}, as well as the condition $\bar{\epsilon}<\gammaunsafe/\Lc(x_i)$ of Proposition \ref{thm:0} may be incompatible, leading to infeasibility of an optimization problem built around them.  This incompatibility arises from the fact that we are simultaneously asking for the value of $\cbf(\xx)$ to vary from $\gammasafe$ to $\gammaunsafe$ over a short distance $\bar{\epsilon}$ while having a low Lipschitz constant.  In particular, as posed, the constraints require that $|h(x_s)-h(x_u)|\geq \gammasafe+\gammaunsafe$ for $x_s \in \Xsafe$ and $x_u\in \XN$ safe and unsafe samples, respectively, but the sampling requirements imply that $\norm{x_s-x_u}_2 \leq \bar{\epsilon}+\epsilon$ for at least some pair $(x_s,x_u)$, which in turn implies that 
\begin{align*}
    L(x_u) \gtrsim \frac{|h(x_s)-h(x_u)|}{\norm{x_s-x_u}_2} \gtrsim \frac{\gammasafe + \gammaunsafe}{\bar{\epsilon}+\epsilon}.
\end{align*}
Thus, if $\gammasafe$ and $\gammaunsafe$ are chosen to be too large, we may exceed the required bound of $\gammaunsafe/\bar{\epsilon}$, and set over which $h(x)\geq 0$ may be undesirably small (i.e., the volume of $\C$ would be too small).  

We address this issue as follows: for fixed $\gammasafe$, $\gammaunsafe$, and $\Lc := \sup_{x_i \in \XN}\Lc(x_i)$, constraint \eqref{eq:safe} is relaxed to 
\begin{align}\label{eq:xsafe}
     h(x_i) \geq \gammasafe \:, \:\: x_i \in \bXsafe,
\end{align}
where now
\begin{equation}\label{eq:bXsafe}
\bXsafe=\left\{x_i \in \Xsafe \, \big{|} \, \inf_{x\in\XN}\|\xx-\xx_i\|_p \ge \frac{\gammaunsafe + \gammasafe}{\Lc} \right\}
\end{equation} 
corresponds to an inner subset of expert trajectory samples.  Intuitively, this introduces a buffer region across which $\cbf(\xx)$ can vary in value from $\gammasafe$ to $-\gammaunsafe$ without having an excessively large Lipschitz constant.  
A near identical argument as that used to prove Proposition \ref{thm:0} can now be used to guarantee that the set $\C$ defined in equation \eqref{eq:local_C} contains the set 
\begin{align*}
    \bD=\bigcup_{x_i \in \bXsafe}\B_{\epsilon,p}(\xx_i),
\end{align*}
defined as the union of $\epsilon$-balls around the points in $\bXsafe$, and thus, $ \bD \subseteq \C$ can be seen as a {``minimum-volume''}
%``minimum-volume''
guarantee on the set $\C$. %\D \subseteq \N \subseteq \Ss$.
\begin{coro}\label{coro:0}
Let $\cbf(\xx)$ be Lipschitz continuous with local constant $L_h(\xx)$.
Let $\gammasafe>0$, and $\Xsafe$ be an ${\epsilon}$-net of $\D$ with ${\epsilon}\leq\gammasafe/\Lc(x_i)$ for all $x_i \in \bXsafe$. Then, if constraint \eqref{eq:xsafe} is satisfied, it holds that $\cbf(\xx)\geq 0$ for all $x\in\bD$.
\end{coro}

\subsection{Guaranteeing valid local  control barrier functions}
\label{sec:theory:valid_local}
The conditions in the previous subsection guarantee that the level-sets of the learned CBF satisfy the desired properties.  We now derive conditions that ensure that the derivative constraint \eqref{eq:cbf_const} is also satisfied by the learned CBF.

Because we assume that the CBF functions $\cbf(\xx)$ are {twice} continuously differentiable over a compact domain $\N\cup\D$, we immediately have that $\nabla\cbf(\xx)$ is Lipschitz continuous with local Lipschitz constant $\Ln(x)$.  Note that to verify that a CBF $\cbf(\xx)$ satisfying the constraints of the previous section is valid, it suffices to show that there exists a single control input $\uu\in\U$ such that the derivative constraint \eqref{eq:cbf_const} holds.  Our approach is to use the control inputs {$\{ u_i : (x_i, u_i) \in \Zdynamics \}$} provided by the expert demonstrations.  We discuss the consequences of this choice further in Section \ref{sec:data}.

To that end, note that for a fixed $\uu_i$, the function $\q(\xx):= \langle \nabla\cbf(\xx),\f(\xx)+\g(\xx)\uu_i\rangle+\alpha(\cbf(\xx))$ is Lipschitz continuous, with Lipschitz constant denoted by $\Lq(\xx)$, as $\nabla h$, $f$ and $g$ are all assumed to be Lipschitz continuous.  Following a similar argument as in the previous subsection, we then have the following result guaranteeing that the learned CBF satisfies the derivative constraint \eqref{eq:cbf_const} for all $x \in \D$.

\begin{proposition}\label{thm:1}
Suppose $q(x)$ is Lipschitz continuous with constant $\Lq(\xx)$.  Let $\gammadynamics>0$, and $\Xsafe$ be an $\epsilon$-net of $\D$ with $\epsilon\leq \gammadynamics/\Lq(\xx_i)$ for all $x_i\in \Xsafe$. Then if
\begin{align}\label{eq:derivative}
\langle \nabla\cbf(\xx_i),\f(\xx_i)+\g(\xx_i)\uu_i\rangle \ge -\alpha(\cbf(\xx_i))+\gammadynamics
\end{align}
for all $x_i \in \Xsafe,$ it holds that $q(x)\geq 0$ for all $x \in \D$.

\end{proposition}
\begin{proof}
Following a similar argument as the proof of Proposition \ref{thm:0}, we note that by equation \eqref{eq:derivative}, we have that $\q(\xx_i)\geq\gammadynamics$ for each $\xx_i\in\Xsafe$. We then have, for any $x\in\D$, that there exists a point $x_i \in \Xsafe$ satisfying $\norm{x-x_i}_p\leq {\epsilon}\leq\gammadynamics/\Lq(x_i)$, from which the following chain of inequalities follows immediately
\begin{align*}
     q(\xx) &=   q(\xx_i) + q(\xx)-q(\xx_i)
      \geq \gammadynamics - |q(\xx)-q(\xx_i)| \\
      &\geq \gammadynamics -\Lq(x_i)\norm{x-x_i}_p  \geq \gammadynamics -\Lq(x_i){\epsilon} \geq 0,
\end{align*}
where the first inequality follows from the assumption that $q(\xx_i)\geq  \gammadynamics$ for all $\xx_i \in \Xsafe$, the second by the Lipschitz assumption on $q(x)$, the third by the assumption that $\Xsafe$ forms an ${\epsilon}$-net of $\Xsafe$, and the final inequality by the condition on $\epsilon$ of the proposition.
\end{proof}

The following theorem, which follows immediately from the previous results, states a set of sufficient conditions guaranteeing that a learned CBF is locally valid on the domain $\N\cup \D$.  We next use these conditions to formulate an optimization based approach to learning a CBF from {expert demonstrations}. %data.
 
\begin{theorem}\label{thm:2}
Let a twice continously differentiable function $\cbf(\xx)$ be a candidate CBF, and let the sets $\Ss$, $\N$, $\D$, $\C$, and $\bD$, and the data-sets $\XN$, $\Xsafe$, and $\bXsafe$ be defined as above. Suppose that $\XN$ forms a $\bar{\epsilon}$-net of $\N$ satisfying the conditions of Proposition \ref{thm:0}, and that $\Xsafe$ forms an $\epsilon$-net of $\D$ satisfying the conditions of Corollary \ref{coro:0} \& Proposition \ref{thm:1}.  Then if $\cbf(\xx)$ satisfies constraints \eqref{eq:unsafe}, \textcolor{red}{\eqref{eq:xsafe}},  and \eqref{eq:derivative}, it holds that the set $\C$ is non-empty, $\bD\subseteq\C\subset\D\subseteq\Ss$, and the function $\cbf(\xx)$ is a valid local control barrier function on $\D$ with domain $\N \cup \D$.

\end{theorem}

\subsection{Control barrier filters} We introduce a simple and natural extension to the notion of a local CBF. Consider the same scenario as above, together with an additional set $\F\subseteq \Ss \setminus \C$ that satisfies the following condition: for each $\zeta_0\in \F$ there exists no continuous signal $\zeta:\RR_{\ge 0}\to\RR^\statedim$ with $\zeta(0):=\zeta_0$ and with $\zeta(t')\not\in \Ss$ for some $t'>0$ and $\zeta(t'')\not\in \C$ for all $t''>0$. This  means that the set $\C$ filters all trajectories starting from $\F$, i.e., each trajectory starting from $\F$ has to pass through $\C$ to escape $\Ss$ and thereby renders $\F$ safe (see Figure \ref{fig:control-barrier-filter}). This follows in the spirit of set invariance \cite{blanchini1999set}.  As illustrated in Figure \ref{fig:control-barrier-filter}, this allows us to remove the perhaps counter-intuitive requirement of having to introduce ``artificial'' unsafe samples in a region that is clearly safe, further reducing the conservatism of the resulting controller.

\subsection{Computing a Control Barrier Function}
\label{sec:compute}
Using the results of the previous subsection, we propose solving the following optimization problem to learn a CBF from expert trajectories:
\begin{subequations} \label{eq:opt}
\begin{alignat}{2}
    &\minimize_{h \in \calH} &&\norm{h} \notag \\
    &\st &&h(x_i) \geq \gammasafe \:, \quad \forall x_i \in \bXsafe(\Lc) \notag \\
    & &&h(x_i) \leq -\gammaunsafe   \notag \\
    & && \Lip(h(x_i),\bar{\epsilon}) \leq \Lc \quad \forall x_i \in \XN \:  \label{eq:lip1}\\
    & &&  q(x_i,u_i):=\langle\nabla h(x_i), f(x_i, u_i) \rangle + \alpha(h(x_i)) \geq \gammadynamics \:  \notag\\
    & && \Lip(q(x_i,u_i),\epsilon) \leq \Lq \quad \forall (x_i, u_i) \in \Zdynamics  \label{eq:lip2}
\end{alignat}
\end{subequations}

The positive constants $\gammasafe$, $\gammaunsafe$, $\gammadynamics$, $\Lc$ and $\Lq$ are
hyperparameters that are set according to the conditions of Theorem \ref{thm:2} given data-sets $\Xsafe$ and $\XN$ defining corresponding $\epsilon$ and $\bar{\epsilon}$-nets.  Here the constraints defined in equations \eqref{eq:lip1} and \eqref{eq:lip2} assume that there exists a function $\Lip(\cdot,\epsilon)$ that returns an upper bound on the Lipschitz constant of its argument in an {$\epsilon$-neighborhood.}
%$\epsilon$ neighborhood.  
We note that it may be difficult to enforce these bounds while solving the optimization problem, in which case we must resort to bootstrapping the values of $\Lc$ and $\Lq$ by iteratively solving optimization problem \eqref{eq:opt}, computing the values $\Lc$ and $\Lq$ for the learned CBF $\cbf(\xx)$, verifying if the conditions of Theorem \ref{thm:2} hold, and readjusting the hyperparameters accordingly if not.  This is a standard approach to hyperparameter tuning, and we show in Section~\ref{sec:experiments} that it can indeed be successfully applied to {verifying} the conditions of Theorem \ref{thm:2}.

\subsubsection{Convexity} 
We first note that optimization problem~\eqref{eq:opt} is convex in $h$ if the function $\alpha$ is linear in its argument, and if we exclude the 
bounds \eqref{eq:lip1} and \eqref{eq:lip2}, and instead verify them via the bootstrapping method described above.  Therefore, if $\calH$ is parameterized as
$\calH = \{ h_\theta(\cdot) = \langle \phi(\cdot), \theta \rangle : \theta \in \Theta \}$
with $\Theta$ a convex set and $\phi(\cdot)$ a known but possibly nonlinear 
transformation, then problem~\eqref{eq:opt} is convex, and can be solved efficiently using standard solvers.
Note that very rich 
function classes such as infinite dimensional RKHS from
statistical learning theory can be approximated to arbitrary accuracy
as such a $\calH$ \cite{rahimi2008random}.

In the more general case when $\calH = \{ h_\theta(\cdot) : \theta \in \Theta \}$,
such as when $h$ is a DNN or when $\alpha$ is a general nonlinear function of its argument, optimization problem~\eqref{eq:opt}
is non-convex. Due to the computational complexity
of general nonlinear constrained programming, we propose an unconstrained
relaxation of problem~\eqref{eq:opt} which can be solved efficiently in practice
by first order gradient based methods. Let $[x]_+ = \max\{x, 0\}$ for $x \in \RR$.
Our unconstrained relaxation {results in} the following optimization problem:
\begin{align}
    \minimize_{\theta \in \Theta} & \norm{\theta}^2 + \lambda_{\mathrm{s}} \sum_{x_i \in \bXsafe} \Big[\gammasafe - h_\theta(x_i)\Big]_+ + \lambda_{\mathrm{u}} \sum_{x_i \in \XN} \Big[h_\theta(x_i) + \gammaunsafe\Big]_+ \label{eq:opt_relaxed} \\
    &\qquad + \lambda_{\mathrm{d}} \sum_{(x_i, u_i) \in \Zdynamics} \Big[\gammadynamics - \Big( \Big\langle \nabla h_\theta(x_i), f(x_i, u_i) \Big\rangle \nonumber + \alpha(h_\theta(x_i)) \Big) \Big]_+  \notag
\end{align}
The positive parameters 
$\lambda_{\mathrm{s}}, \lambda_{\mathrm{u}},\lambda_{\mathrm{d}}$ allow us to trade off the relative importance of each of the 
terms in the optimization.
While equation~\eqref{eq:opt_relaxed} is in general a non-convex optimization problem,
it can be solved efficiently in practice with stochastic first-order gradient methods
such as Adam or SGD.

\subsubsection{Lipschitz continuity of $\calH$}
As described earlier, because we assume that functions in $\calH$ are twice continuously differentiable and we restrict ourselves to a compact domain $\N \cup \D$, we immediately have that
$h$ and $\nabla h$ are both uniformly Lipschitz over $\N \cup \D$.
We show here two examples of $\calH$ where
it is computationally efficient to estimate an upper bound on the
Lipschitz constants of functions $h\in\calH$.

In the case of random Fourier features with $\ell$ random features, where $h(x) = \ip{\phi(x)}{\theta}$ and $\phi(x) \in \R^\ell$ is
\begin{align*}
    \phi(x) = \sqrt{\frac{2}{\ell}} (\cos(\ip{x}{w_1} + b_1), \dots , \cos(\ip{x}{w_\ell} + b_\ell)) \:,
\end{align*}
then we can analytically compute
upper bounds as follows.
First, we have by the Cauchy-Schwarz inequality $\abs{h(x_1) - h(x_2)} \leq \norm{\phi(x_1) - \phi(x_2)}_2\norm{\theta}_2$.
To bound $ \norm{\phi(x_1) - \phi(x_2)}_2$, we bound the spectral norm of the Jacobian
$D\phi(x)$, which is a matrix where the $i$-th
row is $-\sqrt{2/\ell} \sin(\ip{x}{w_i} + b_i) w_i^\T$.
Let $s_i := \sin(\ip{x}{w_i} + b_i)$ and 
observe that
\begin{align*}
    \norm{D\phi(x)} = \sqrt{\frac{2}{\ell}} \sup_{\norm{v}_2=1} \left( \sum_{i=1}^{\ell} s_i^2 \ip{w_i}{v}^2 \right)^{1/2} \leq \sqrt{\frac{2}{\ell}} \sup_{\norm{v}_2=1} \left( \sum_{i=1}^{\ell} \ip{w_i}{v}^2\right)^{1/2} = \sqrt{\frac{2}{\ell}} \norm{W}, 
\end{align*}
where $W$ is a matrix with the $i$-th row equal to $w_i$. While the bound
$\sqrt{2/\ell}\norm{W}$ can be used in computations,
we can further understand order-wise scaling
of the bound as follows.
For random Fourier features corresponding to the popular Gaussian radial basis function kernel, $w_i \stackrel{\mathrm{iid}}{\sim} \mathsf{N}(0, \sigma^2 I)$ where $\sigma^2$ is the (inverse) bandwidth of the Gaussian kernel.  Therefore, by standard results in non-asymptotic random matrix theory \cite{vershynin2018high}, we have that
\begin{align*}
    \norm{W} \leq \sigma(\sqrt{\ell} + \sqrt{n} + \sqrt{2\log(1/\delta)})
\end{align*}
w.p. at least $1-\delta$. Combining these calculations, we have that the Lipschitz constant of $h$ can be bounded by $\sqrt{2\sigma^2} (1 + \sqrt{n/\ell} + \sqrt{(2/\ell) \log(1/\delta)}) \norm{\theta}_2$ w.p. at least $1-\delta$.

We now bound the Lipschitz constant of the gradient $\nabla h(x) = D\phi(x)^\T \theta$. We do this by bounding the spectral norm of the 
Hessian $\nabla^2 h(x)= -\sqrt{2/\ell} \sum_{i=1}^{\ell} c_i \theta_i w_iw_i^\T$, with $c_i = \cos(\ip{x}{w_i} + b_i)$.
A simple bound is
\begin{align*}
    \norm{\nabla^2 h(x)} \leq \sqrt{2/\ell} \norm{\theta}_\infty \norm{W}^2 \leq 3\sqrt{2}\norm{\theta}_\infty \sigma^2 (\ell+n+2\log(1/\delta))/\sqrt{\ell},
\end{align*}
where the last inequality holds w.p. at least $1-\delta$.

When $h(x)$ is a DNN, accurately estimating the Lipschitz constant
is more involved.  In general, the problem of exactly computing the Lipschitz constant of $h$ is known to be NP-hard \cite{virmaux2018lipschitz}.  Notably, because most commonly-used activation functions $\phi$ are known to be 1-Lipschitz (e.g. ReLU, tanh, sigmoid), a naive upper bound on the Lipschitz constant of $h$ is given by the product of the norms of the weight matrices; that is, $L_h \leq \prod_k\norm{W^k}$.  However, this bound is known to be quite loose~\cite{fazlyab2019efficient}.  Recently, the authors of \cite{fazlyab2019efficient} proposed a semidefinite-programming based approach to efficiently compute an accurate upper bound on $L_h$.  In particular, this approach relies on incremental quadratic constraints to represent the couplings between pairs of neurons in the neural network $h$. On the other hand, there are relatively few results that provide accurate upper bounds for the Lipschitz constant of the gradient of $h$ when $h$ is a neural network.  While ongoing work looks to extend the results from \cite{fazlyab2019efficient} to compute upper bounds on $\Lip(\nabla h)$, to the best of our knowledge, the only general method for computing an upper bound on $\Lip(\nabla h)$ is through post-hoc sampling \cite{wood1996estimation}.

\subsection{Data Collection}
\label{sec:data}
We briefly comment on how data should be collected to ensure that the conditions of Theorem \ref{thm:2} are satisfied.

\subsubsection*{What should the experts do?} At a high level, our results state that if a smooth CBF can be found that satisfies the constraints \eqref{eq:safe}, \eqref{eq:unsafe}, and \eqref{eq:derivative} over a sufficiently fine sampling of the state-space, then the resulting function is a valid CBF. We focus here on the derivative constraint \eqref{eq:derivative}, which must be verified to hold for \emph{some} $u \in \U$, by using the expert example data $(x_i,u_i)$.  In particular, the more transverse the vector field $f(x_i,u_i)$ is to the level sets of the learned CBF $h(x_i)$ {(i.e., the more parallel it is to the inward pointing normal $\nabla h(x_i)$)}, the larger the inner-product term in constraint \eqref{eq:derivative} is \emph{without} increasing the Lipschitz constant of $h(x)$.  In words, this says that the expert demonstrations \emph{should demonstrate how to move away from the unsafe set.}  {This also highlights the role of actuation authority in the ability to learn smooth CBFs: systems with larger actuation authority are more easily  able to align the closed loop vector field $f(x_i,u_i)$ away from the unsafe set.}

\subsubsection*{Constructing $\epsilon$-nets}
In order to construct an $\epsilon$-net of a set $\Ss$,
a simple randomized algorithm which repeatedly uniformly samples
from $\Ss$ works with high probability (see, for example, \cite{vershynin2018high}). Hence, as long as we can efficiently sample from $\Ss$ (e.g. when $\Ss$ is a basic primitive set or has an efficient set-membership oracle), uniform sampling is a viable strategy.  Alternatively, a gridding approach can be taken.  We note that in either case, for a set of diameter $r$ on the order of $O(\left(\frac{r}{\epsilon}\right)^d)$ samples are required.  While this exponential dependence is undesirable, we observe that in practice, the expert demonstrations allow us to focus on a subset of the state-space associated with desirable behavior, significantly reducing the diameters of the sets to be sampled.

\section{Numerical Experiments}
\label{sec:experiments}

{All code is publicly available at \href{https://github.com/unstable-zeros/learning-cbfs}{https://github.com/unstable-zeros/learning-cbfs}}.

\subsection{Planar Example}
\label{sec:experiments:planar}

Our first experiment is the
following two dimensional planar system
adapted from \cite{kolathaya19}:
\begin{align}
    \dot{x}_1 &= - x_1 + (x_1^2 + \delta) u_1  \\
    \dot{x}_2 &= - x_2 + (x_2^2 + \delta) u_2, \notag 
\end{align}
where $\delta>0$ is a fixed parameter guaranteeing that the
system is globally feedback linearizable. We set $\delta = 1$
in our experiments.
The desired safe set is $\Ss = \{ x : x_1 \leq 1 \:, \:\: x_2 \leq 1 \}$. We generate expert data for this system as follows.
Because the system is feedback linearizable,
given a desired trajectory $x_d(t)$, we can
easily design a nominal controller which tracks $x_d(t)$.
We can then construct a safe controller (w.r.t. $\Ss$)
by solving the CBF-QP problem \cite{ames2014control,ames2017control} with 
the  CBF $h(x) = \min\{1-x_1, 1-x_2\}$.

We design two sets of desired trajectories.
Let the unit vector $v(\theta) = [-\cos{\theta} \ \sin{\theta}]^\T$.
The first set is defined for a fixed $r > 0$ as $x_d(t) = r v(t)$ from $t \in [0, 2\pi]$.
We do this for $r \in \{0.2666, 0.3, 0.3333\}$,
sampling $80$ time equi-spaced points along each curve.
The second set of desired trajectories are
for a fixed $\theta \in [0, 2\pi]$, where we
consider a trajectory that starts at 
$x(0) = 0.4666v(\theta)$ and
ends up at $x(t_f) = 0.3666v(\theta)$,
and one where $x(0) = 0.1333v(\theta)$
and $x(t_f) = 0.2333v(\theta)$.
We grid across both $\theta \in [0, 2\pi]$ and $t \in [0, t_f]$ to ensure a densely sampled
set of points. All sample points $\Zdynamics$ are shown in Figure~\ref{fig:toy_data}(left). We consider the $x_i$ corresponding to
the circular trajectories (green in Fig.~\ref{fig:toy_data}) as defining $\bXsafe$.
We then set $\Xunsafe$ to be points sampled
(red in Fig.~\ref{fig:toy_data}) along the
circle at $r=-0.5$ and $r=-0.1$.
Our samples are specifically chosen to
form a net over $\D$
and $\N$, 
with $\epsilon =0.01666$ and $\bar{\epsilon} = 0.0333$, respectively.

\begin{figure}[t!]
\centering
~\includegraphics[height=1.75in]{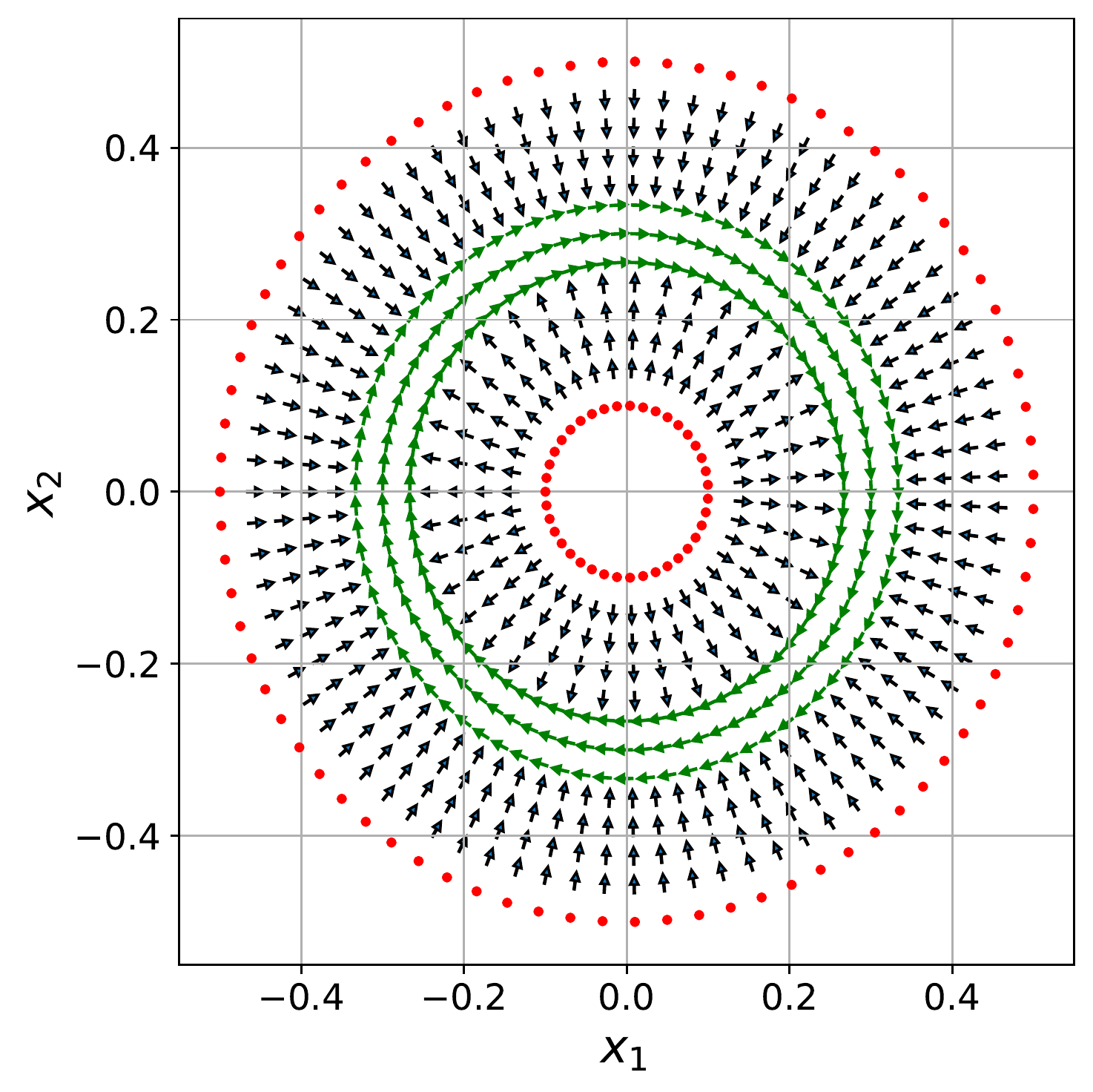}~\includegraphics[height=1.75in]{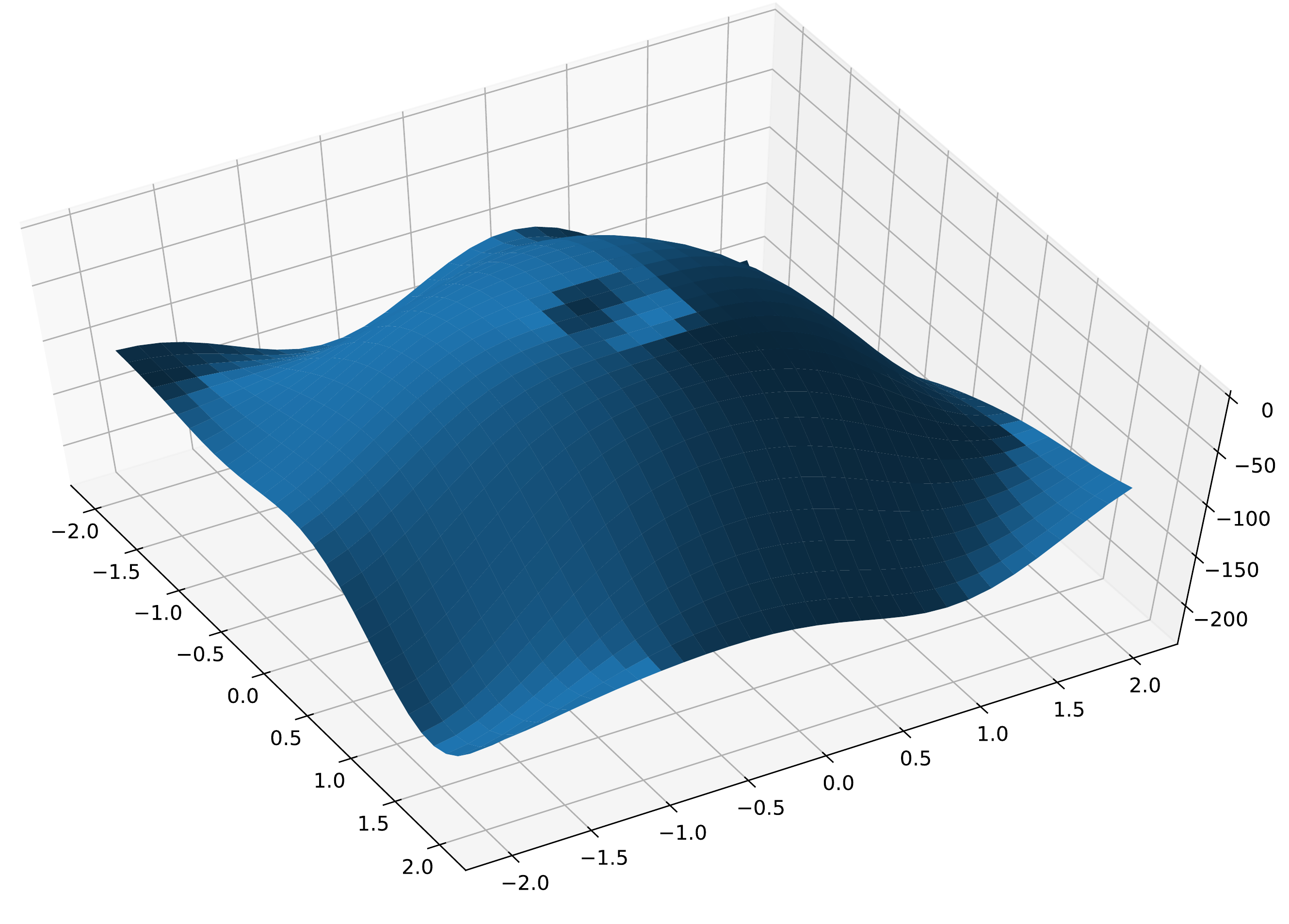}~\includegraphics[height=1.75in]{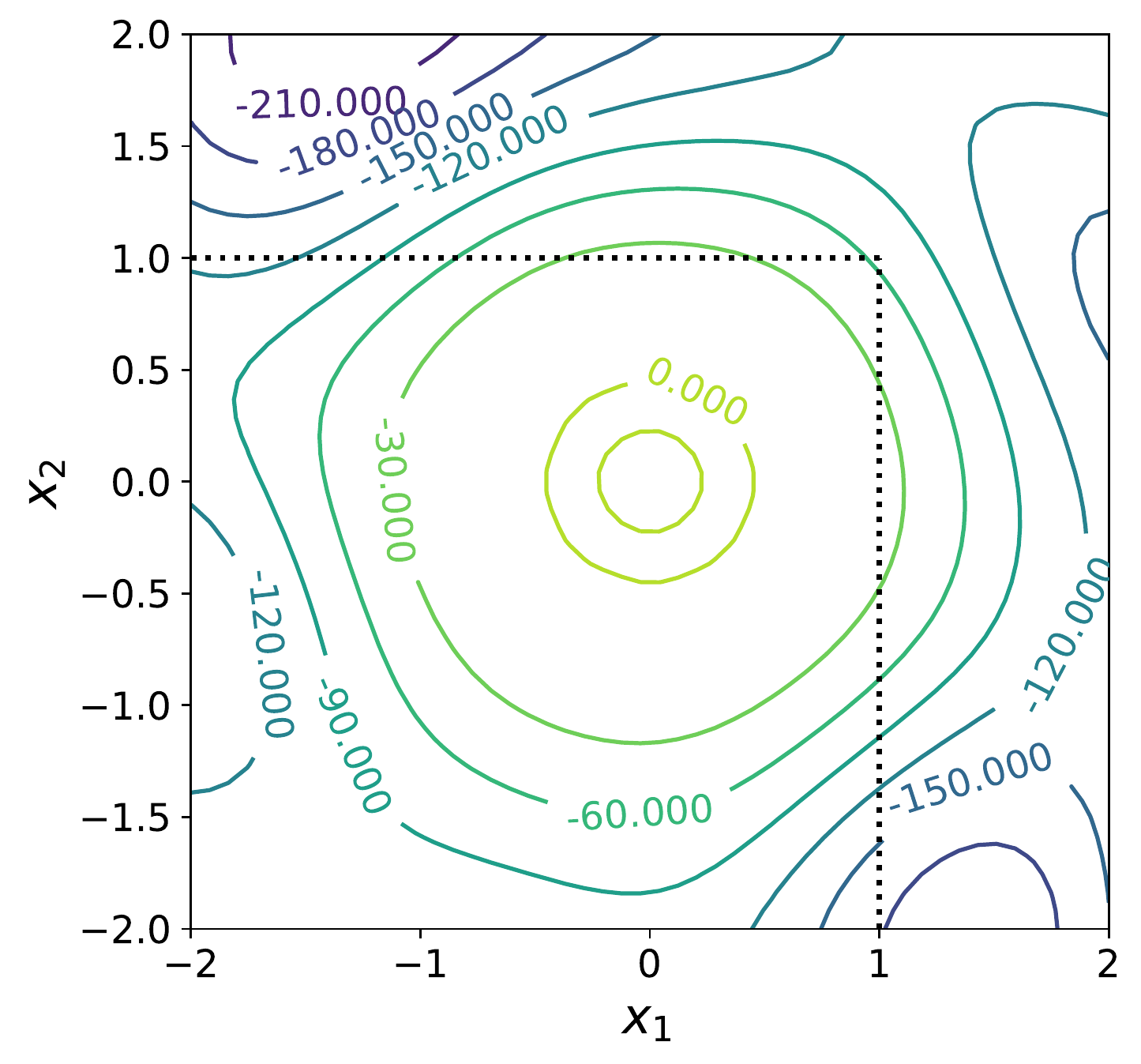}~
\caption{Left: Plot of the expert trajectories (green), dynamic samples (black) and unsafe samples (red) used for training.  Center: Surface plot of the learned CBF. Right: Level set plot of the learned CBF for the two dimensional planar example. The dotted black line
represents the boundary of the safe set $\Ss$.}
\label{fig:toy_data}
\end{figure}

We parameterize $\mathcal{H}$ using $\ell=200$
random Fourier features corresponding to the
Gaussian kernel with $\sigma = 1.2$. We set $\alpha(x) = x$ and then solve the optimization problem
with $\gammasafe = 0.1$,
$\gammaunsafe = 0.3$,
$\gammadynamics = 0.01$
using cvxpy \cite{cvxpy} with the MOSEK \cite{mosek} backend. 
Next, we verify that our specific choices of
$\gammasafe, \gammaunsafe, \gammadynamics$
satisfied the necessary conditions,
computing $\norm{\nabla h(x_i)}$ and
$\norm{\nabla_x q(x_i, u_i)}$ to obtain
$L_h(x_i)$ and $L_{q}(x_i)$, respectively. This verification is shown in Fig.~\ref{fig:toy_slacks}.
{The resulting CBF $h(x)$ is plotted in Fig.~\ref{fig:toy_data} (center), and its  level sets are shown in Fig.~\ref{fig:toy_data}(right), from which it can be seen that the zero level sets are well within the safe set (demarcated by the black dotted line).  As can be observed, the set $\C$ is an annulus with approximate inner radius of $.2333$ and and approximate outer radius of $.4$, and we note that the corresponding set $\D$ over which the CBF is valid is an annulus with inner radius of approximate radius $.1333$ of approximate radius $.4666$.} {Finally, in Fig.~\ref{fig:trajectories}, we show the evolution of a system governed by the CBF-QP controller defined by the learned CBF $h(x)$ for circular reference trajectories of varying radii $r$, beginning at initial conditions of $x_0 = (-r,0)$.  First observe that for the trajectory of radius $r=0.3$, which lies within the CBF safe set $\C$, we replicate the expert behavior (dashed orange line).  Next, notice that all other trajectories are seen to converge to the learned safe set $\C$ -- perhaps surprisingly, even those trajectories beginning well outside of the set $\D$ over which the learned CBF is provably valid exhibit this favorable behavior, suggesting that the smoothness conditions imposed during training allow for generalization well beyond previously seen expert behavior.}

\begin{figure}
\centering
    \includegraphics[width=\columnwidth]{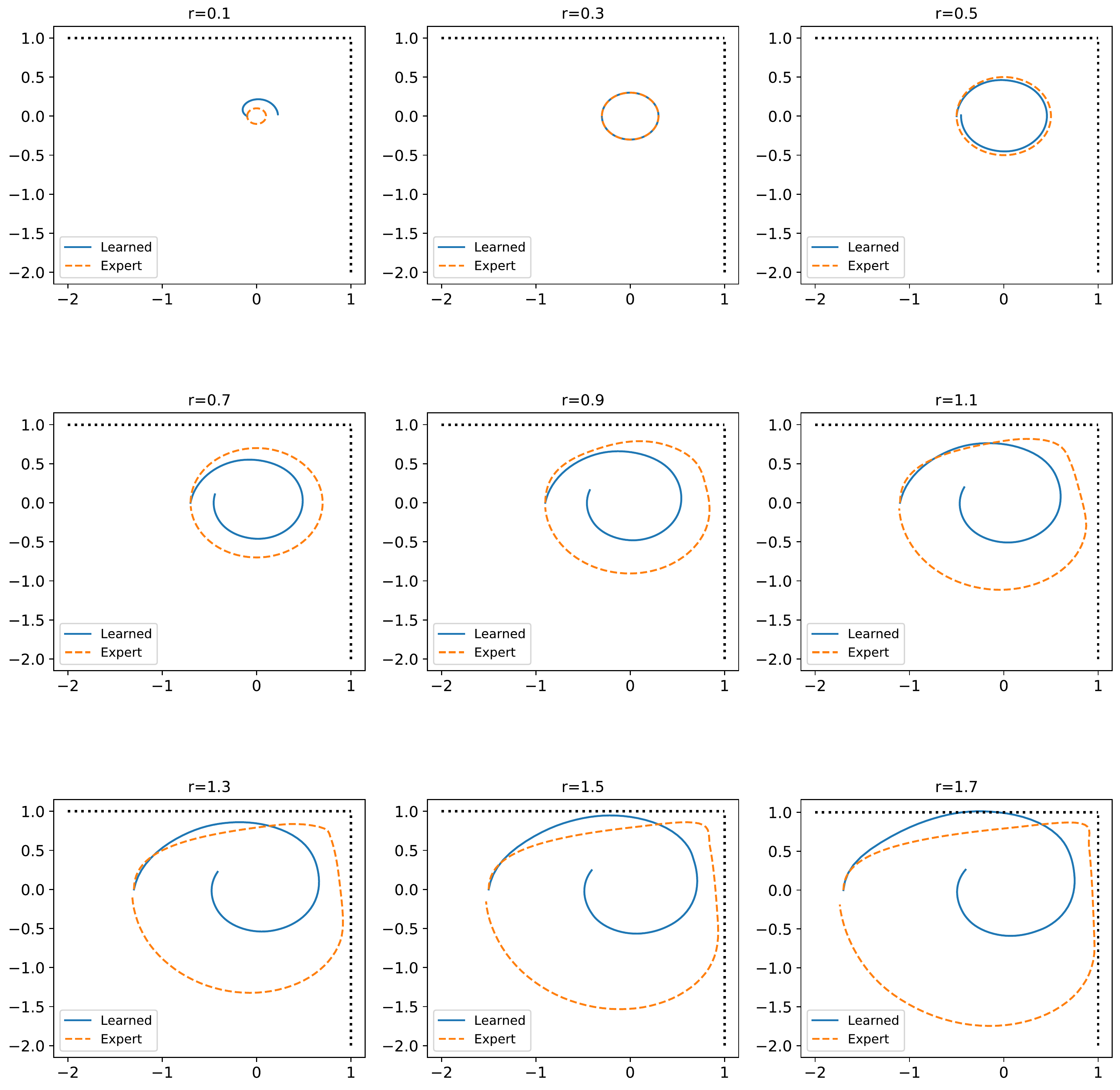}
    \caption{System trajectories for the planar example under the CBF-QP controller defined by the learned CBF $h$ shown in Fig.~\ref{fig:toy_data}.  For initial conditions beginning outside of $\mathcal{C}=\{ x \, | \, h(x) \geq 0\}$, note that the trajectory converges to $\mathcal{C}$, illustrating the robustness benefits of the set $\mathcal{D}\supset \mathcal{C}$.}
    \label{fig:trajectories}
\end{figure}

% \begin{figure}[h]
% \centering

% \caption{}
% \label{fig:toy_level_set}
% \end{figure}

\begin{figure}[ht]
    \centering
    \includegraphics[width=0.8\columnwidth]{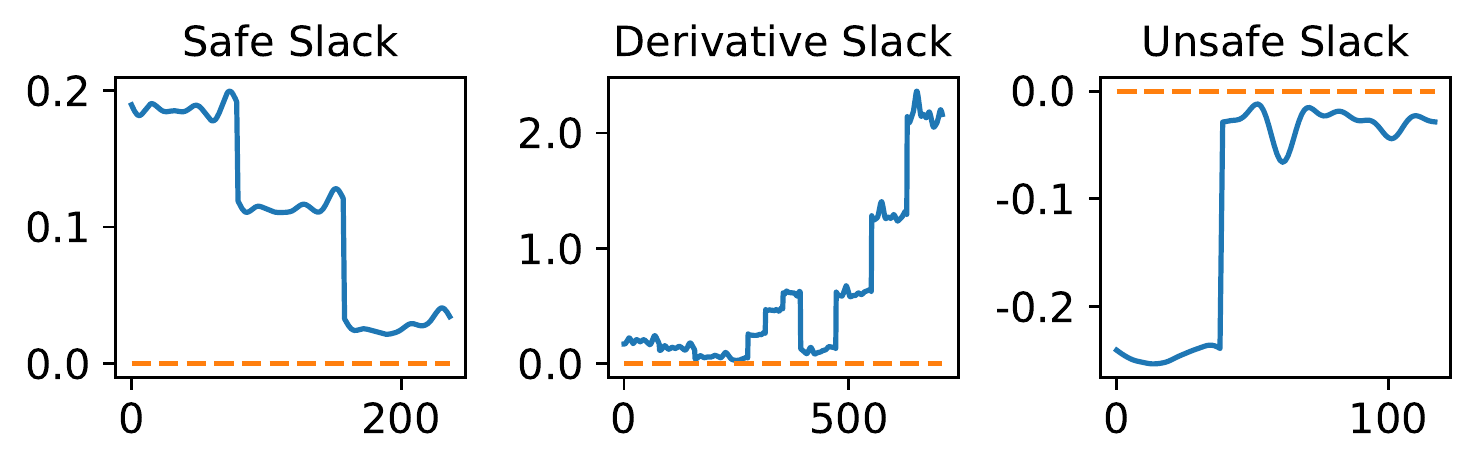}
    \caption{Safe, derivative, and unsafe slacks which verify the sufficient conditions given in Theorem~\ref{thm:2}.  (Left) The safe slack plot, for each constraint, shows the value $h(x_i) - L_h(x_i) \epsilon$, which needs to be positive.  (Middle) The derivative slack plot, for each constraint, shows the value of $q(x_i, u_i) - L_q(x_i) \epsilon$, which also needs to be positive.  (Right) The unsafe slack plot, for each constraint, shows the value of $h(x_i) + L_h(x_i) \bar{\varepsilon}$, which needs to be negative.}
    \label{fig:toy_slacks}
\end{figure}

\subsection{Aircraft Collision Avoidance}
\label{sec:experiments:uav}

In this subsection, we apply the control barrier filter technique (Section \ref{sec:theory:valid_local}) to the aircraft collision avoidance problem in \cite{squires2018constructive}. 
The joint state vector of the two aircraft, indexed with $a$ and $b$, is $x=[p_{x,a} \ p_{y,b} \ \theta_{a} \ p_{x,b} \ p_{y,b} \ \theta_{b}]^\T \in \mathbb{R}^6$, denoting positions in the $(x,y)$-plane and orientations.
% We discretized the model with sampling time $t_s = 0.1$s using the Runge–Kutta 4th order method.
% \stephen{(Where do we need model discretization? 
% This is potentially confusing to the reader.)}
The controls $u=[v_{a} \ \omega_{a} \ v_{b} \ \omega_{b}]^\T \in \mathbb{R}^4$ are the translational and angular velocities with constraints $0.1 \leq v_{a} \leq 1.0$ and $-1.0 \leq \omega_{a} \leq 1.0$.
The control goal is to reach the target states $x_g=[-5 \ 0 \ \pi \ 5 \ 0 \ 0]^\T$ if $p_{x,a}(0) \geq 0$, $p_{x,b}(0) \leq 0$ or $x_g=[5 \ 0 \ 0 \ -5 \ 0 \ \pi]^\T$ if $p_{x,a}(0) \leq 0$, $p_{x,b}(0) \geq 0$.
The safety specification is that the two airplanes should maintain a minimal distance $D_s = 0.5$ to avoid collisions. 
To this end, we define the geometric safe set as,
\begin{equation}
\label{eq::airplane_safe_set}
\begin{aligned}
\Ss:= \left\{x \in \R^6 \,\middle\vert\, p_{x,r}^{2}+p_{y,r}^{2} \geq D_{s}^{2} \right\}.
\end{aligned}
\end{equation}
where $p_{i,r} = p_{i,a}-p_{i,b}, \ i \in \{x,y\}$ is the relative position.

\subsubsection*{Generating training data}
We consider two ways of generating expert demonstrations.
First, we used a standard tracking model predictive contoller (MPC) as the nominal controller equipped with the closed form constructive CBF in~\cite{squires2018constructive} for collision avoidance (which we refer to as CBF-MPC).  To generate the expert trajectories, we started the system from 400 randomly generated initial conditions inside the set $\Ss$.  Each run terminated when the airplanes were sufficiently far away from each other.  Furthermore, we {obtained safe and unsafe samples by uniformly sampling from the sets $\F = \left\{x \in \R^6 \,\middle\vert\, (3 D_s)^2 \leq p_{x,r}^{2}+p_{y,r}^{2} \leq (5 D_s)^{2} \right\}$ and $\N = \left\{x \in \R^6 \,\middle\vert\, p_{x,r}^{2}+p_{y,r}^{2} \leq (1.1 D_{s})^{2} \right\}$, respectively}; these state samples, as well as the expert trajectories, are shown in Figure \ref{fig:data_plane} in relative coordinates.

Secondly, we built a web-based simulator that allows a user to control two simulated aerial vehicles.  As before, the goal of the simulation is to control the two aerial vehicles such that they do not collide.  We emphasize that these trajectories were solely by human guidance; no nominal controller was used. The data is plotted in Figure \ref{fig:data_plane}.

\begin{figure}[t]
\centering
~\includegraphics[scale=0.7]{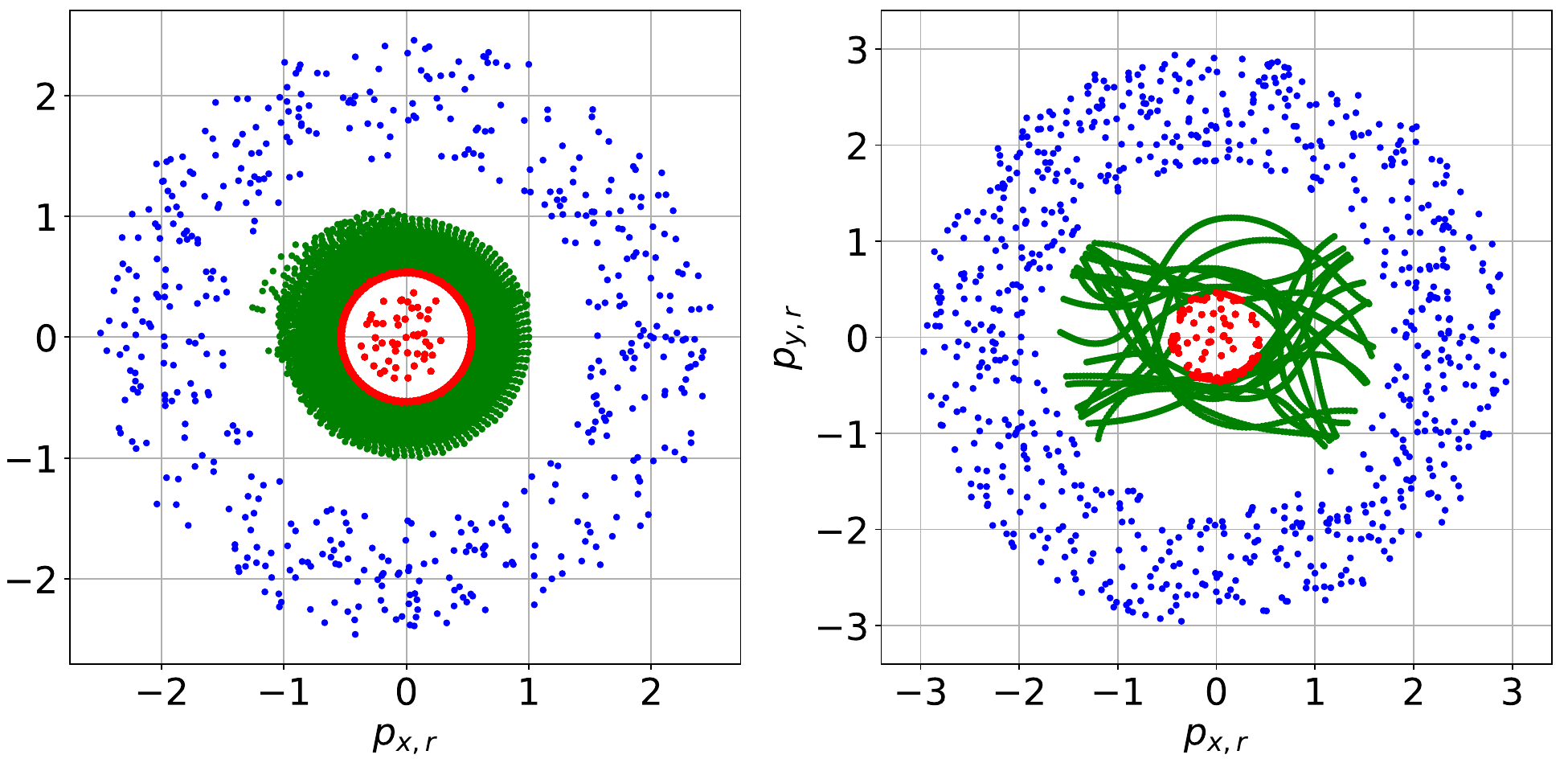}~
\caption{Left: Plot of the expert trajectories $\Zdynamics$ generated by CBF-MPC (green), safe samples $\XS$ (blue) and unsafe samples $\XN$ (red). 
Right: Plot of the expert trajectories obtained from human demonstrations with the safe and unsafe samples.} 
\label{fig:data_plane}
\end{figure}

\subsubsection*{Training procedure}
We parametrized the CBF candidate $h(x)$ with a two-hidden-layer fully-connected neural network with 64 neurons in each layer and tanh activation functions.
The training procedure was implemented using JAX \cite{jax2018github} and the Adam algorithm with a cosine decay learning rate.
We trained the neural network for $10^5$ epochs using the loss in \eqref{eq:opt_relaxed} with $\lambda_{\mathrm{s}} = 2.0$, $\lambda_{\mathrm{u}} = 2.0$, $\lambda_{\mathrm{d}} = 15.0$, $\gamma_{\text{safe}} = 3.0$, $\gamma_{\text{unsafe}} = 0.5$ and $\gamma_{\text{dyn}} = 0.05$.  
Each of these hyperparameters was chosen via grid-search.
%$\gamma_{\text{safe,near}} = 0.15$, $ \gamma_{\text{unsafe,near}} = 0.025$,
The learned CBFs and the closed form CBF~\cite{squires2018constructive} evaluated at the training points are plotted in Figure \ref{fig:cbf_plane} in relative coordinates.

\begin{figure}[t!]
\centering
~\includegraphics[scale=0.7]{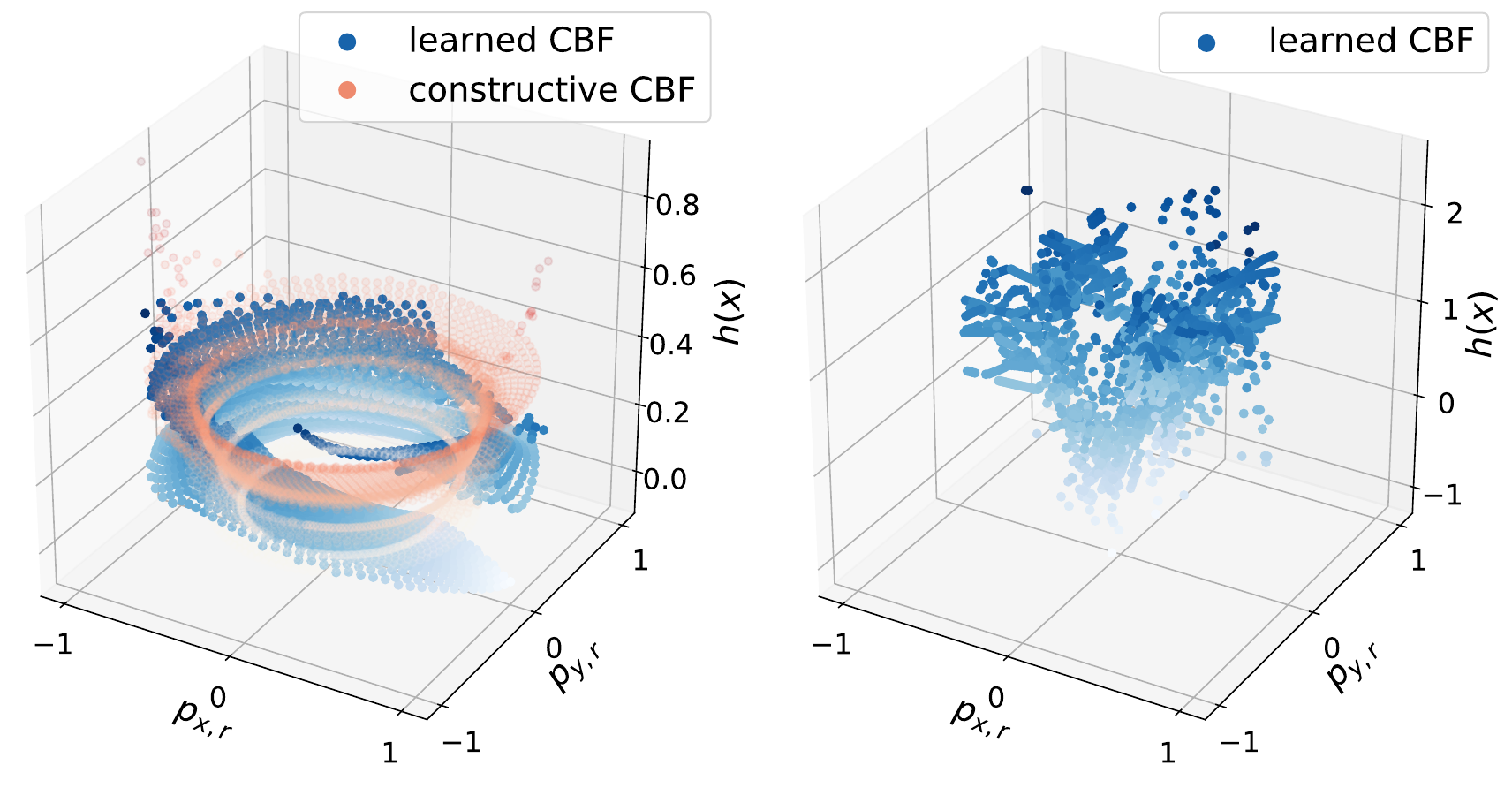}
\caption{Left: The CBF learned from CBF-MPC (blue) and the closed form CBF in~\cite{squires2018constructive} (red) evaluated at the states in training data-set. Right: The CBF learned from human demonstrations. }
\label{fig:cbf_plane}
\end{figure}

\subsubsection*{Closed-loop control with learned CBF}  To demonstrate the efficacy of the CBF learned from expert demonstrations, we used it in the aircraft collision avoidance problem with the same control goal and safety specification as in \eqref{eq::airplane_safe_set}.  The two airplanes were initialized at various symmetric initial positions on the circle $p_x^2+p_y^2 = 1$ such that they were facing each other.  In this way, if both airplanes used the nominal MPC controller, they would collide.  

The closed-loop state trajectories using our learned CBF are shown in Figure \ref{fig:closed_loop_plane}.  The CBFs learned on both data-sets successfully steer the airplanes away from each other for all initial states, which experimentally validates the forward invariance of $\Ss$.  As a comparison, we also plotted the state trajectories produced by the CBF from~\cite{squires2018constructive} under the same settings in Figure \ref{fig:closed_loop_plane}.
Since this CBF is derived analytically, it appears to render more aggressive control actions which manage to separate the airplanes at a closer distance.

\begin{figure}[t!]
\centering
\includegraphics[scale=0.45]{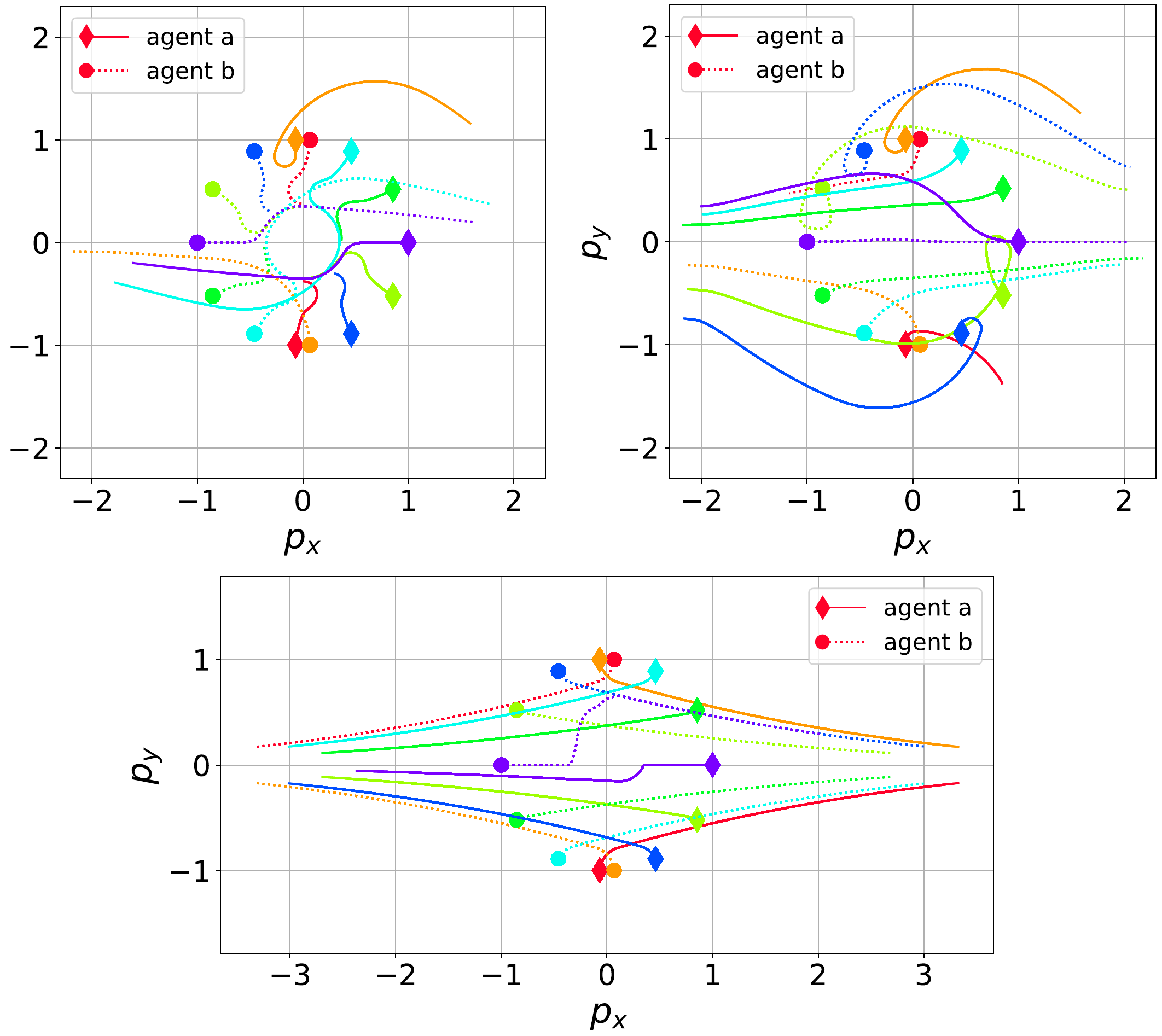}
\caption{Closed-loop control of the two airplanes starting from multiple initial conditions using closed form CBF from~\cite{squires2018constructive} (top left), the CBF learned from CBF-MPC (top right), and the CBF learned from human demonstrations (bottom). Trajectories of the same run are marked with an identical color. The initial states of agent a are marked with diamonds, and those of agent b with circles.}
\label{fig:closed_loop_plane}
\end{figure}

\section{Conclusion}
\label{sec:conclusion}
We proposed and analyzed an optimization based approach to learning CBFs from expert demonstrations for known nonlinear control affine dynamical systems.  We showed that under suitable assumptions of smoothness on the underlying dynamics and the learned CBF (which can be guaranteed using classic and recent \cite{fazlyab2019efficient} results for RKHS and DNNs), and under sufficiently fine sampling, the learned CBF is provably valid, guaranteeing safety.  This work provides a firm theoretical foundation for future exploration that will look to leverage tools from statistical learning theory to reduce the sample complexity burden of the proposed method by focusing on guaranteeing safety for ``typical'' behaviors, as opposed to uniform coverage of the state-space.

\section{Acknowledgements}

This work was supported in part by the Knut and Alice Wallenberg Foundation (KAW) and the SSF COIN project.

%\newpage 
\bibliographystyle{IEEEtran}
\bibliography{main}

\begin{thebibliography}{10}
\providecommand{\url}[1]{#1}
\csname url@samestyle\endcsname
\providecommand{\newblock}{\relax}
\providecommand{\bibinfo}[2]{#2}
\providecommand{\BIBentrySTDinterwordspacing}{\spaceskip=0pt\relax}
\providecommand{\BIBentryALTinterwordstretchfactor}{4}
\providecommand{\BIBentryALTinterwordspacing}{\spaceskip=\fontdimen2\font plus
\BIBentryALTinterwordstretchfactor\fontdimen3\font minus
  \fontdimen4\font\relax}
\providecommand{\BIBforeignlanguage}[2]{{%
\expandafter\ifx\csname l@#1\endcsname\relax
\typeout{** WARNING: IEEEtran.bst: No hyphenation pattern has been}%
\typeout{** loaded for the language `#1'. Using the pattern for}%
\typeout{** the default language instead.}%
\else
\language=\csname l@#1\endcsname
\fi
#2}}
\providecommand{\BIBdecl}{\relax}
\BIBdecl

\bibitem{prajna2007framework}
S.~Prajna, A.~Jadbabaie, and G.~J. Pappas, ``A framework for worst-case and
  stochastic safety verification using barrier certificates,'' \emph{IEEE
  Trans. Autom. Control}, vol.~52, no.~8, pp. 1415--1428, 2007.

\bibitem{parrilo2000structured}
P.~A. Parrilo, ``Structured semidefinite programs and semialgebraic geometry
  methods in robustness and optimization,'' Ph.D. dissertation, California
  Institute of Technology, 2000.

\bibitem{wieland2007constructive}
P.~Wieland and F.~Allg{\"o}wer, ``Constructive safety using control barrier
  functions,'' in \emph{Proc. IFAC Symp. Nonlin. Control Syst.}, Pretoria,
  South Africa, August 2007, pp. 462--467.

\bibitem{ames2014control}
A.~D. Ames, J.~W. Grizzle, and P.~Tabuada, ``Control barrier function based
  quadratic programs with application to adaptive cruise control,'' in
  \emph{Proc. Conf. Decis. Control}, Los Angeles, CA, December 2014, pp.
  6271--6278.

\bibitem{ames2017control}
A.~D. Ames, X.~Xu, J.~W. Grizzle, and P.~Tabuada, ``Control barrier function
  based quadratic programs for safety critical systems,'' \emph{IEEE Trans.
  Autom. Control}, vol.~62, no.~8, pp. 3861--3876, 2017.

\bibitem{sontag1989universal}
E.~D. Sontag, ``A 'universal' construction of artstein's theorem on nonlinear
  stabilization,'' \emph{Systems \& control letters}, vol.~13, no.~2, pp.
  117--123, 1989.

\bibitem{xu2015robustness}
X.~Xu, P.~Tabuada, J.~W. Grizzle, and A.~D. Ames, ``Robustness of control
  barrier functions for safety critical control,'' in \emph{Proc. Conf. Analys.
  Design Hybrid Syst.}, Atlanta, GA, October 2015, pp. 54--61.

\bibitem{ames2019control}
A.~D. Ames, S.~Coogan, M.~Egerstedt, G.~Notomista, K.~Sreenath, and P.~Tabuada,
  ``Control barrier functions: Theory and applications,'' in \emph{Proc.
  European Control Conf.}, Naples, Italy, June 2019, pp. 3420--3431.

\bibitem{xu2017correctness}
X.~Xu, J.~W. Grizzle, P.~Tabuada, and A.~D. Ames, ``Correctness guarantees for
  the composition of lane keeping and adaptive cruise control,'' \emph{IEEE
  Transactions on Automation Science and Engineering}, vol.~15, no.~3, pp.
  1216--1229, 2017.

\bibitem{wang2018permissive}
L.~Wang, D.~Han, and M.~Egerstedt, ``Permissive barrier certificates for safe
  stabilization using sum-of-squares,'' in \emph{Proc. American Control Conf.},
  Milwaukee, WI, June, pp. 585--590.

\bibitem{cheng2019end}
R.~Cheng, G.~Orosz, R.~M. Murray, and J.~W. Burdick, ``End-to-end safe
  reinforcement learning through barrier functions for safety-critical
  continuous control tasks,'' in \emph{Proc. Conf. Artificial Intel.},
  Honolulu, HI, February 2019, pp. 3387--3395.

\bibitem{wang2018safe}
L.~Wang, E.~A. Theodorou, and M.~Egerstedt, ``Safe learning of quadrotor
  dynamics using barrier certificates,'' in \emph{Proc. Conf. Robot. Automat.},
  Brisbane, Australia, May 2018, pp. 2460--2465.

\bibitem{taylor2019learning}
A.~Taylor, A.~Singletary, Y.~Yue, and A.~Ames, ``Learning for safety-critical
  control with control barrier functions,'' in \emph{Proc. Conf. Learning for
  Dynamics and Control}, June 2020, pp. 708--717.

\bibitem{taylor2020control}
A.~J. Taylor, A.~Singletary, Y.~Yue, and A.~D. Ames, ``A control barrier
  perspective on episodic learning via projection-to-state safety,''
  \emph{arXiv preprint arXiv:2003.08028}, 2020.

\bibitem{yaghoubi2020training}
S.~Yaghoubi, G.~Fainekos, and S.~Sankaranarayanan, ``Training neural network
  controllers using control barrier functions in the presence of
  disturbances,'' \emph{arXiv preprint arXiv:2001.08088}, 2020.

\bibitem{jin2020neural}
W.~Jin, Z.~Wang, Z.~Yang, and S.~Mou, ``Neural certificates for safe control
  policies,'' \emph{arXiv preprint arXiv:2006.08465}, 2020.

\bibitem{boffi2020learning}
N.~M. Boffi, S.~Tu, N.~Matni, J.-J.~E. Slotine, and V.~Sindhwani, ``Learning
  stability certificates from data,'' \emph{arXiv preprint arXiv:2008.05952},
  2020.

\bibitem{srinivasan2020synthesis}
M.~Srinivasan, A.~Dabholkar, S.~Coogan, and P.~Vela, ``Synthesis of control
  barrier functions using a supervised machine learning approach,'' \emph{arXiv
  preprint arXiv:2003.04950}, 2020.

\bibitem{saveriano2019learning}
M.~Saveriano and D.~Lee, ``Learning barrier functions for constrained motion
  planning with dynamical systems,'' in \emph{Proc. Conf. Intelligent Robots
  Systems}, Macau, China, November 2019.

\bibitem{squires2018constructive}
E.~Squires, P.~Pierpaoli, and M.~Egerstedt, ``Constructive barrier certificates
  with applications to fixed-wing aircraft collision avoidance,'' in
  \emph{Proc. Conf. Control Techn. Appl.}, Copenhagen, Denmark, August 2018,
  pp. 1656--1661.

\bibitem{rosolia2017learning}
U.~Rosolia and F.~Borrelli, ``Learning model predictive control for iterative
  tasks. a data-driven control framework,'' \emph{IEEE Trans. Autom. Control},
  vol.~63, no.~7, pp. 1883--1896, 2017.

\bibitem{rahimi2008random}
A.~Rahimi and B.~Recht, ``Random features for large-scale kernel machines,'' in
  \emph{Proc. Advances Neur. Inform. Proc. Syst.}, Vancouver, Canada, December
  2008, pp. 1177--1184.

\bibitem{fazlyab2019efficient}
M.~Fazlyab, A.~Robey, H.~Hassani, M.~Morari, and G.~Pappas, ``Efficient and
  accurate estimation of lipschitz constants for deep neural networks,'' in
  \emph{Proc. Advances Neur. Inform. Proc. Syst.}, Vancouver, Canada, December
  2019, pp. 11\,423--11\,434.

\bibitem{blanchini1999set}
F.~Blanchini, ``Set invariance in control,'' \emph{Automatica}, vol.~35,
  no.~11, pp. 1747--1767, 1999.

\bibitem{vershynin2018high}
R.~Vershynin, \emph{High-dimensional prob.: An introduction with applications
  in data science}.\hskip 1em plus 0.5em minus 0.4em\relax Cambridge university
  press, 2018, vol.~47.

\bibitem{virmaux2018lipschitz}
A.~Virmaux and K.~Scaman, ``Lipschitz regularity of deep neural networks:
  analysis and efficient estimation,'' in \emph{Proc. Advances Neur. Inform.
  Proc. Syst.}, Montréal, Canada, December 2018, pp. 3835--3844.

\bibitem{wood1996estimation}
G.~Wood and B.~Zhang, ``Estimation of the lipschitz constant of a function,''
  \emph{Journ. of Global Opt.}, vol.~8, no.~1, pp. 91--103, 1996.

\bibitem{kolathaya19}
S.~{Kolathaya} and A.~D. {Ames}, ``Input-to-state safety with control barrier
  functions,'' \emph{IEEE Control Systems Letters}, vol.~3, no.~1, pp.
  108--113, 2019.

\bibitem{cvxpy}
S.~Diamond and S.~Boyd, ``{CVXPY}: A {P}ython-embedded modeling language for
  convex optimization,'' \emph{Journal of Machine Learning Research}, vol.~17,
  no.~83, pp. 1--5, 2016.

\bibitem{mosek}
\BIBentryALTinterwordspacing
M.~ApS, \emph{The MOSEK optimization toolbox for MATLAB manual. Version 9.0.},
  2019. [Online]. Available: \url{http://docs.mosek.com/9.0/toolbox/index.html}
\BIBentrySTDinterwordspacing

\bibitem{jax2018github}
\BIBentryALTinterwordspacing
J.~Bradbury, R.~Frostig, P.~Hawkins, M.~J. Johnson, C.~Leary, D.~Maclaurin, and
  S.~Wanderman-Milne, ``{JAX}: composable transformations of {P}ython+{N}um{P}y
  programs,'' 2018. [Online]. Available: \url{http://github.com/google/jax}
\BIBentrySTDinterwordspacing

\end{thebibliography}

\end{document}